\documentclass[11pt,a4paper]{article} 

\usepackage[a4paper,hmargin=2.8cm,vmargin=3cm]{geometry}
\usepackage[utf8x]{inputenc} 
\usepackage{amsmath,amsthm,amsfonts,amssymb}
\usepackage{authblk}
\usepackage[bookmarksnumbered=true]{hyperref} 
\usepackage{dsfont} 
\usepackage{color} 
\usepackage{slashed}
\usepackage{graphicx}  
\usepackage{caption}
\usepackage{tikz-cd}
\usetikzlibrary{positioning}
 
\usepackage[capitalise]{cleveref}
\crefname{equation}{}{}
\Crefname{lem}{Lemma}{Lemmas}   
   
\newcommand\pgfmathsinandcos[3]{%
  \pgfmathsetmacro#1{sin(#3)}%
  \pgfmathsetmacro#2{cos(#3)}%
}
\newcommand\LongitudePlane[3][current plane]{%
  \pgfmathsinandcos\sinEl\cosEl{#2} 
  \pgfmathsinandcos\sint\cost{#3} 
  \tikzset{#1/.style={cm={\cost,\sint*\sinEl,0,\cosEl,(0,0)}}}
}

\newcommand\LatitudePlane[3][current plane]{%
  \pgfmathsinandcos\sinEl\cosEl{#2} 
  \pgfmathsinandcos\sint\cost{#3} 
  \pgfmathsetmacro\yshift{\RadiusSphere*\cosEl*\sint}
  \tikzset{#1/.style={cm={\cost,0,0,\cost*\sinEl,(0,\yshift)}}} %
}

\newcommand\DrawLongitudeArc[4][black]{
  \LongitudePlane{\angEl}{#2}
  \tikzset{current plane/.prefix style={scale=1}}
  \pgfmathsetmacro\angVis{atan(sin(#2)*cos(\angEl)/sin(\angEl))} %
  \pgfmathsetmacro\angA{mod(max(\angVis,#3),360)} %
  \pgfmathsetmacro\angB{mod(min(\angVis+180,#4),360} %
  \draw[current plane,#1]  (\angA:\RadiusSphere) arc (\angA:\angB:\RadiusSphere);
}%
\newcommand\DrawLatitudeCircle[2][1]{
  \LatitudePlane{\angEl}{#2}
  \tikzset{current plane/.prefix style={scale=1}}
  \pgfmathsetmacro\sinVis{sin(#2)/cos(#2)*sin(\angEl)/cos(\angEl)}
  \pgfmathsetmacro\angVis{asin(min(1,max(\sinVis,-1)))}
  \draw[current plane] (\angVis:\RadiusSphere) arc (\angVis:-\angVis-180:\RadiusSphere);
}

\newcommand\DrawLatitudeArc[4][black]{
  \LatitudePlane{\angEl}{#2}
  \tikzset{current plane/.prefix style={scale=1}}
  \pgfmathsetmacro\sinVis{sin(#2)/cos(#2)*sin(\angEl)/cos(\angEl)}
  \pgfmathsetmacro\angVis{asin(min(1,max(\sinVis,-1)))}
  \pgfmathsetmacro\angA{max(min(\angVis,#3),-\angVis-180)} %
  \pgfmathsetmacro\angB{min(\angVis,#4)} %

  \draw[current plane,#1] (\angA:\RadiusSphere) arc (\angA:\angB:\RadiusSphere);
}

\newtheorem{thm}{Theorem}[section]

\newtheorem{lem}[thm]{Lemma}

\theoremstyle{remark}
\newtheorem*{rem}{Remark}
\newtheorem*{rems}{Remarks}

\numberwithin{equation}{section}

\newcommand{\fock}{\mathcal{F}}		
\newcommand{\di}{{\textnormal{d}}}		
\newcommand{\rfrak}{\mathfrak{r}}

\newcommand{\Tbb}{\mathbb{T}}
\newcommand{\W}{W}
\newcommand{\Wt}{\widetilde{W}}
\newcommand{\D}{D}

\newcommand{\Ecal}{\mathcal{E}}
\newcommand{\Hbb}{\mathbb{H}}
\newcommand{\Fcal}{\mathcal{F}}
\newcommand{\Ncal}{\mathcal{N}}		
\newcommand{\Hcal}{\mathcal{H}}		
\newcommand{\Ical}{\mathcal{I}}
\newcommand{\Rcal}{\mathcal{R}}
\newcommand{\Ikp}{\Ical_{k}^{+}}
\newcommand{\Ikm}{\Ical_{k}^{-}}
\newcommand{\Ik}{\Ical_{k}}
\newcommand{\Il}{\Ical_{l}}

\newcommand{\ik}{I}

\newcommand{\Dfrak}{\mathfrak{D}}

\newcommand{\Ocal}{\mathcal{O}}		
\newcommand{\hc}{\textnormal{h.c.}}		

\newcommand{\Rbb}{\mathbb{R}}		
\newcommand{\Cbb}{\mathbb{C}}		
\newcommand{\Nbb}{\mathbb{N}}		
\newcommand{\Zbb}{\mathbb{Z}}

\newcommand{\Xbb}{\mathbb{X}}

\newcommand{\Dbb}{\mathbb{D}}

\renewcommand{\Re}{\operatorname{Re}} 	
\newcommand{\id}{\mathbb{I}}
\newcommand{\norm}[1]{\lVert#1\rVert}	

\newcommand{\tr}{\operatorname{tr}}
\newcommand{\HS}{_{\textnormal{HS}}}

\newcommand{\sgn}{\operatorname{sgn}}

\newcommand{\BFc}{B_\F^c}
\newcommand{\BF}{B_\F}

\newcommand{\Efrak}{\mathfrak{E}}

\newcommand{\Kfrak}{\mathfrak{K}}

\newcommand{\kappaf}{\kappa}

\newcommand{\north}{\Gamma^{\textnormal{nor}}}

\newcommand{\F}{\textnormal{F}} 
\newcommand{\B}{\textnormal{B}} 

\newcommand{\diag}{\operatorname{diag}}

\newcommand{\supp}{\operatorname{supp}}

\newcommand{\linb}{\textnormal{B}}

\title{Bosonization of Fermionic Many--Body Dynamics} 

\author[1]{Niels Benedikter}
\author[2]{Phan Th\`anh Nam}
\author[3]{Marcello Porta} 
\author[4]{Benjamin Schlein} 
\author[5]{Robert Seiringer} 
\affil[1]{Universit\`a degli Studi di Milano, Dipartimento di Matematica, Via Cesare Saldini 50, 20133 Milano, Italy}
\affil[2]{LMU Munich, Department of Mathematics, Theresienstra{\ss}e 39, 80333 M\"unchen, Germany} 
\affil[3]{SISSA, Mathematics Area, Via Bonomea 265, 34136 Trieste, Italy}
\affil[4]{Institute of Mathematics, University of Zurich, Winterthurerstrasse 190, 8057 Zurich, Switzerland}
\affil[5]{IST Austria, Am Campus 1, 3400 Klosterneuburg, Austria}

\begin{document} 

\maketitle
\abstract{We consider the quantum many--body evolution of a homogeneous Fermi gas in three dimensions in the coupled semiclassical and mean-field scaling regime. We study a class of initial data describing collective particle--hole pair excitations on the Fermi ball. Using a rigorous version of approximate bosonization, we prove that the many--body evolution 
can be approximated in Fock space norm by a quasifree bosonic evolution of the collective particle--hole excitations.}

\tableofcontents

\section{Introduction} 
The problem of computing quantum correlations in fermionic many--body systems has a long history in theoretical physics. A widely used nonperturbative method is the \emph{random phase approximation} (RPA), introduced by Bohm and Pines \cite{BP53}. Despite its popularity, the mathematical validity of this approach remained elusive until recently. One of the earliest applications of the RPA concerns the correlation energy of interacting fermionic systems at high density, defined as the difference between many--body and Hartree--Fock ground state energies. The RPA allows to derive a nonperturbative expression for the correlation energy, a prediction already contained in the foundational paper \cite{BP53}. In the mean--field regime, the validity of this expression at second order in the interaction potential has been proved in \cite{HPR20}. The first nonperturbative justification of the RPA for the correlation energy in the mean--field regime has been established in \cite{BNPSS, BNPSS2}. The result was extended in \cite{CHN21, BPSS21}.

The key concept of our approach is to interpret certain delocalized pairs of fermions as bosons with an effective quadratic Hamiltonian, making it possible to compute the ground state energy using a Bogoliubov transformation to diagonalize the effective Hamiltonian. In the present paper, we develop this approach further and derive a norm approximation for the fermionic many--body quantum dynamics in terms of an effective bosonic dynamics, generated by a quadratic Hamiltonian. In particular, our result identifies a class of almost stationary states that are associated with the excited eigenvalues of the many--body Hamiltonian of the fermionic system.

\subsection{Fermi Gas in the Mean--Field Scaling Regime} We consider a system of $N$ spinless fermionic particles on the torus $\Tbb^3 := \Rbb^3/(2\pi \Zbb^3)$. The dynamics is governed by the Schr\"odinger equation 
\begin{equation} \label{eq:intro-Sch}
i\hbar \partial_t \Psi_N(t) = H_N \Psi_N(t)
\end{equation}
where the Hamiltonian has the form
\begin{equation} \label{eq:HN}
H_N := \hbar^2 \sum_{i=1}^N \left( - \Delta_{x_i} \right) + \lambda \sum_{1 \leq i < j \leq N} V\left( x_i - x_j \right)\;,
\end{equation}
and the wave function $\Psi_N(t)$ belongs to the space of antisymmetric functions 
\begin{equation} \label{eq:intro-antisymm_space}
L^2_\textnormal{a}(\Tbb^{3N}) := \{ \psi \in L^2( (\Tbb^{3})^N) : \psi(x_{\sigma(1)}, \ldots, x_{\sigma(N)}) = \sgn(\sigma) \psi(x_1,\ldots,x_N)\ \forall \sigma \in \mathcal{S}_N \}\;.
\end{equation}
Here $\mathcal{S}_N$ is the group of permutations of $N$ symbols. 
We assume that the Fourier transform $\hat{V}: \Zbb^3 \to \Rbb$ of the interaction potential $V$ is non-negative and compactly supported. In this case, $H_N$ is bounded from below and its self--adjointness follows from the Kato--Rellich theorem. Consequently, by Stone's theorem, the solution of \eqref{eq:intro-Sch} for any initial wave function $\Psi_N(0) \in L^2_\textnormal{a}(\Tbb^{3N})$ is given by $\Psi_N(t) = e^{-it H_N/\hbar}\Psi_N(0)$. 

We are interested in the behavior of the system when $N\to \infty$ in the \emph{coupled semiclassical and mean--field scaling regime} 
\begin{align} \label{eq:choice-hbar-lambda}
\hbar \simeq N^{-\frac{1}{3}}\;, \quad \lambda := N^{-1}\;. 
\end{align}
(To be precise, in the next paragraph we will define $\hbar$ in terms of the Fermi momentum $k_\F$.)
In this case, for typical low--energy wave functions, the kinetic energy and the interaction energy are both of order $N$. This scaling regime was considered by \cite{NS81,Spo81,CLS21,CLL21} for the derivation of the Vlasov equation and by \cite{EESY04,BPS16,BPS14b,BPS14,PRSS17} for the derivation of the Hartree--Fock equation. Different scaling limits have been considered in \cite{BGGM03, BGGM04, FK11,PP16,BBPPT16}. Note that the convergences in these works are mostly concerned the one--body density matrices, which are in principle less precise than the norm approximation. 

\paragraph{Hartree--Fock approximation.} To leading order, physical properties of weakly interacting fermionic systems can often be approximated by Hartree--Fock theory; see, e.g., \cite{Bac92,GS94} for the ground state energy and  the papers just cited in the previous paragraph for the dynamics. In Hartree--Fock theory one restricts the Hilbert space of antisymmetric wave functions to its submanifold of \emph{Slater determinants} $\Psi_\textnormal{Sl} = \bigwedge_{j=1}^N \varphi_j$ with $\varphi_j \in L^2(\Tbb^3)$, i.\,e., antisymmetrized elementary tensors. Since Slater determinants are quasi--free states, the Wick theorem can then be used to obtain the Hartree--Fock energy functional
\[\begin{split}\Ecal_\textnormal{HF}(\gamma) & := \langle \Psi_\textnormal{Sl}, H_N \Psi_\textnormal{Sl}\rangle \\
& = \tr (-\hbar^2\Delta \gamma) + \frac{1}{2N}\int \di x \di y V(x-y) \left( \gamma(x;x) \gamma(y;y) - \lvert \gamma(x;y) \rvert^2 \right)\end{split}\]
depending only on the one--particle reduced density matrix
\[\gamma(x;y) := \int \di x_2 \cdots \di x_N \Psi_\textnormal{Sl}(x,x_2,\ldots,x_N) \overline{\Psi_\textnormal{Sl}(y,x_2,\ldots,x_N)}\;.\]

In general, the analysis of the Hartree--Fock variational problem (the minimization of $\Ecal_\textnormal{HF}$ over the set of one--particle wave function $\varphi_j$) is still not trivial. Therefore we will assume that the particle number $N$ is such that it fills completely the \emph{Fermi ball}
\begin{equation}
\BF := \{ k \in \Zbb^3: \lvert k\rvert \leq k_\F\}\;,
\end{equation}
i.\,e., we let the number of particles be $N := \lvert \BF \rvert$. This simplifies the Hartree--Fock problem for the translation invariant Hamiltonian \eqref{eq:HN}, in the coupled semiclassical and mean--field scaling regime: namely, the minimum of the Hartree--Fock functional is given by plane waves \cite[Appendix A]{BNPSS2} as in the non--interacting case:
\begin{equation}
\psi_\textnormal{pw} := \bigwedge_{k \in B_\F} e_k\;, \qquad e_k(x) := (2\pi)^{-\frac{3}{2}} e^{ik\cdot x} \textnormal{ with } k \in \Zbb^3,\ x \in \Tbb^3\;. \label{eq:plane-waves}
\end{equation}
 To realize the limit of large particle number we then take $k_\F \to +\infty$. According to Gauss' classic counting argument
\begin{equation} \label{eq:latticecounting}
 k_\F =  \kappaf N^{\frac{1}{3}} + \Ocal(1) \qquad \textnormal{where} \quad \kappaf := \left(\frac{3}{4\pi}\right)^{\frac{1}{3}}\;.
\end{equation} 
(The correction term is actually much smaller than $\Ocal(1)$ if we employ advanced number--theoretic results \cite{Hea99} on lattice point counting.)
We define
\begin{equation}
 \hbar := \frac{\kappaf}{k_\F}= N^{-\frac{1}{3}} + \Ocal(N^{-2/3}). 
\end{equation}
(In earlier papers \cite{BNPSS,BNPSS2}, we took $\hbar = N^{-1/3}$. The advantage of the present definition of the scaling is that we have exactly $\kappaf = \left(3/4\pi\right)^{\frac{1}{3}}$ instead of $\kappaf = \left(3/4\pi\right)^{\frac{1}{3}} + \Ocal(N^{-\frac{1}{3}})$, avoiding additional trivial error terms in the effective Hamiltonian.)

Note that the Slater determinant $\psi_\textnormal{pw}$ minimizes the kinetic energy, neglecting the many--body interaction. The only quantum correlations taken into account by this state are those induced by the antisymmetry requirement. To get corrections to Hartree--Fock theory, for example to the ground state energy, we have to go beyond the plane waves ansatz and include non-trivial quantum correlations.
A first step in this direction has been taken in \cite{HPR20}, where the  correction to the ground state energy has been computed to second order in the interaction. To all orders in the interaction, the dominant nonperturbative correction has been obtained in \cite{BNPSS,BNPSS2,BPSS21} via a rigorous collective bosonization method (and by a non--collective bosonization method in \cite{CHN21}). Similar methods have also been used recently in the context of dilute Fermi gases \cite{FGHP21}.
In the next subsection we recall the collective bosonization approach, on which also our new result is based. For this purpose we recall first the formalism of second quantization.  

\paragraph{Second quantization.}
It is convenient to work with creation and annihilation operators, even though we only consider systems with fixed particle number. On fermionic Fock space $\fock$, constructed over $L^2(\Tbb^3)$, we introduce the usual fermionic operators $a^*_p$ creating a particle with momentum $p \in \Zbb^3$, and correspondingly the annihilation operators $a_p$. They satisfy the  \emph{canonical  anticommutation relations} (CAR)
\begin{equation}\label{eq:car}
\{a_p,a^*_{q}\} = \delta_{p,q}\;, \quad \{a_p,a_q\} =0 = \{a^*_p,a^*_q\}\;, \qquad \forall p,q \in \Zbb^3\;. 
\end{equation}
We frequently use the operator norm bounds $\norm{ a^{*}_{p}} \leq 1$ and $\norm{ a_{p}} \leq 1$, which represent the fact that a fermionic mode can be occupied by at most one particle.
The fermionic number operator is $\Ncal = \sum_{p \in \Zbb^3} a^*_p a_p$, and the vacuum vector is denoted by $\Omega$.   
We extend the Hamiltonian to Fock space as
\begin{equation}\label{eq:HcalN}
\Hcal_N := \hbar^2 \sum_{p\in \mathbb{Z}^3}   \lvert p\rvert^2 a^*_p a_p  + \frac{1}{2N} \sum_{k,p,q\in \mathbb{Z}^3} \hat V(k) a^*_{p+k} a^*_{q-k} a_{q} a_{p}\;. 
\end{equation}
Restricted to $L^2_\textnormal{a}((\Tbb^{3})^N) \subset \fock$, $\Hcal_N$ agrees with the $N$--particle Hamiltonian $H_{N}$. In particular, the ground state energy can be written as
\begin{align}
E_{N} & := \textnormal{inf spec} (H_{N}) = \inf_{\psi \in \fock:\,\Ncal\psi = N\psi} \frac{\langle \psi, \mathcal{H}_{N} \psi\rangle}{\langle \psi, \psi \rangle}\;.
\end{align}

\subsection{Correlation Hamiltonian}
It is convenient to start the analysis by employing a particle--hole transformation, which allows us to describe all states in Fock space relative to the non--interacting Fermi ball by creating particles outside or holes inside the Fermi ball. We then have to keep track only of these excitations. The  \emph{particle--hole transformation} is defined as the map $R: \fock \to \fock$ satisfying (in terms of the plane waves $e_p$ introduced in \cref{eq:plane-waves}) 
\begin{equation}
R^* a^*_p R := \left\{ \begin{matrix}{} a^*_p& \textnormal{for } p \in \BFc\\a_p & \textnormal{for } p \in \BF \end{matrix}\right.\;, \qquad  R\Omega := \bigwedge_{p \in \BF} e_p\;.
\label{eq:phtrafo}
\end{equation}
This map is well--defined since the set of all vectors of the form $\prod_{j}a^*_{k_j} \Omega$ forms a basis of $\fock$. Moreover, it is easy to verify that  $R = R^* = R^{-1}$; in particular $R$ is a unitary transformation. (In fact, $R$ is a Bogoliubov transformation, i.\,e., it transforms creation operators into a linear combination of creation and annihilation operators such that the CAR are preserved.) Thus, the energy of the Fermi ball of non-interacting particles is
\begin{equation}\label{eq:pw}
E^{\textnormal{pw}}_{N} = \langle R\Omega, \mathcal{H}_{N} R\Omega\rangle\;.
\end{equation}
An important role in our analysis is played by the \emph{correlation Hamiltonian}
\begin{equation}\label{eq:Hcorr}
\Hcal_\textnormal{corr}:= R^* \Hcal_N R - E^{\textnormal{pw}}_{N} =  \Hbb_0 + Q_\B + \mathcal{E}_1 + \mathcal{E}_2 + \Xbb
\end{equation}
where the terms relevant for the statement of our main result are the kinetic energy and the bosonizable interaction terms,
\begin{align}
\mathbb{H}_{0} &= \sum_{p\in \mathbb{Z}^{3}} e(p) a^{*}_{p} a_{p}\;, \qquad e(p) = \lvert \hbar^{2} \lvert p\rvert^{2} - \kappaf^2\rvert\;,\nonumber\\
Q_\textnormal{B} &= \frac{1}{N} \sum_{k\in \north} \hat V(k) \Big[ b^{*}(k) b(k) + b^{*}(-k) b(-k) + b^{*}(k) b^{*}(-k) + b(-k) b(k) \Big]\;. \label{eq:Hcorrdet}
\end{align}
The summand $\Xbb$ is the exchange term in the effective Hamiltonian of Hartree--Fock theory, $\Ecal_1$ is a summand containing only non--bosonizable contributions, and $\Ecal_2$ couples non--bosonizable contributions to the bosonizable $b$-- and $b^*$--operators. We will show that these three summands are small error terms; since they are not necessary to state our result, we give the precise formulas only where needed, in \cref{sec:red-bos}.
The operator $b^{*}(k)$ is the particle--hole pair creation operator
\begin{equation} \label{eq:global_pair}
b^*(k) := \sum_{p \in \BFc \cap (\BF +k)}  a^*_p a^*_{p-k}\;.
\end{equation}
We have also introduced the set $\north $ of all momenta $k=(k_1,k_2,k_3) \in \Zbb^3 \cap  \supp \hat V$ satisfying  
$k_3> 0 \textnormal{ or } (k_3=0 \textnormal{ and } k_2>0) \textnormal{ or } (k_2=k_3=0 \textnormal{ and } k_1>0)$.
It is chosen such that
$\north  \cap (- \north ) =\emptyset, \quad \north  \cup (- \north ) =\Big( \mathbb{Z}^3 \cap \supp \hat V \Big) \setminus \{0\}$.

\paragraph{Patch decomposition of the Fermi surface.} The operators $b(k)$ and $b^{*}(k)$ annihilate and create, respectively, a pair of fermions delocalized in a shell around the Fermi surface. Because of the delocalization over many 
fermionic states, the Pauli principle is usually negligible for the collective modes generated by these operators. In fact, we can prove that, on states with few excitations of the Fermi ball, the operators $b(k)$ and $b^* (k)$ satisfy (up to appropriate normalization constants) approximately bosonic commutation relations, i.\,e.,
\begin{equation}
[ b(k), b(l) ] = 0\;,\qquad [ b(k), b^{*}(l) ] \simeq \delta_{k,l} \times (\textnormal{normalization}). 
\end{equation}
Notice that it is crucial to have a large number of summands in the definition  \cref{eq:global_pair}: we find $(b^*(k))^m = 0$ only when $m \in \Nbb$ is larger than the number of summands. (A different approach has been recently proposed in \cite{CHN21}, considering, instead of (1.16), operators of the form $a_p^* a_{p-k}^*$; in this case bosonic behavior is only recovered after averaging over $p$; see also Remark~\ref{rem:v} after \cref{thm:1}.)

Furthermore, $Q_\B$ is quadratic in terms of these operators, so that we may try to diagonalize it by a bosonic Bogoliubov transformation. Unfortunately the kinetic energy $\Hbb_0$ does not have an obvious quadratic representation in terms of the $b$-- and $b^*$--operators, providing an obstacle for the application of bosonic Bogoliubov theory. To overcome this problem and express also $\mathbb{H}_{0}$ quadratically in terms of almost--bosonic operators, we need to linearize the dispersion relation $e(p)$ near the Fermi surface (these steps will be explained further in \cref{eq:Hlinearization} and \cref{eq:Dbb}). First we observe that in  \cref{eq:Hcorrdet} we always have $k\in \supp \hat{V}$, and in $b^*(k)$ we have $p$ outside the Fermi ball but $p-k$ inside the Fermi ball; therefore only fermionic operators not further than a distance $\operatorname{diam}\supp \hat V$ from the Fermi surface appear. We then take this shell and decompose the $b$-- and $b^*$--operators into localized operators $b_\alpha$ and $b^*_\alpha$ covering the Fermi surface shell; we call the localization regions \emph{patches}. The important properties of the patch decomposition are that the patches should be separated by thin corridors so that there is no interaction between neighbours, and they should not degenerate into very elongated shapes (the number of points in their interior should be much larger than the number of points near the surrounding corridor). The precise form of the patch decomposition is not relevant as long as these properties are satisfied. As an example, it can be constructed by placing a disc at the north pole, then cutting along the lines of northern latitude, cutting the obtained rings, and finally reflecting by the origin to the southern half sphere. \Cref{fig:blub} illustrates such a patch decomposition of the northern half of the Fermi sphere; patches are then  reflected at the origin to the southern half. We refer to \cite[Section 4]{BNPSS2} for the details and recall only the main aspects in the following:

The number of patches $M$ is a parameter depending on the particle number $N$, and will eventually be optimized in the range 
\[
N^{2\delta} \ll M \ll N^{\frac{2}{3} -2\delta}, \quad 0<\delta< \frac{1}{6}\;,
\] 
where $\delta$ is another parameter independent of $N$ to be optimized at the end of the proof; its main role is to define the patch cut--off around the equator in \cref{eq:cut--off}.
(We need $M \gg 1$ to control the linearization error in \cref{lem:boskin}; the stricter lower bound $M \gg N^{2\delta}$ is required for validity of the counting argument \cref{eq:doku}, as illustrated in \cite[Fig.~2]{BPSS21}. The condition $M \ll N^{\frac{2}{3} -2\delta}$ ensure that the patches contain a large enough number of fermionic modes, required for suppressing the Pauli principle and justifying the neglect of error terms in the approximate CCR.)
The patches $\{ B_{\alpha} \}_{\alpha = 1}^{M}$ have thickness $R_V := \operatorname{diam}\supp \hat V$ in the radial direction, and their side lengths are of equal order. They are non--overlapping and separated by corridors of width strictly larger than $2R_V$. By $\hat \omega_{\alpha}$ we denote the normalized vectors pointing in the direction of the patch centers.

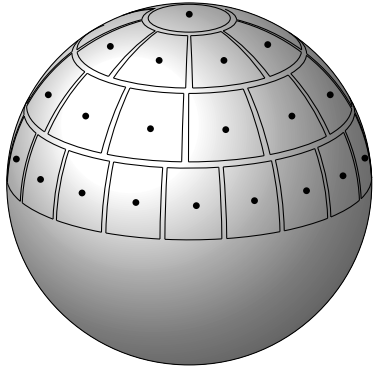
\begin{figure}\centering
\begin{tikzpicture}[scale=0.6]
\def\RadiusSphere{4} 
\def\angEl{20} 
\def\angAz{-20} 

\filldraw[ball color = white] (0,0) circle (\RadiusSphere);

\DrawLatitudeCircle[\RadiusSphere]{75+2}
\foreach \t in {0,-50,...,-250} {
  \DrawLatitudeArc{75}{(\t+50-4)*sin(62)}{\t*sin(62)}
 \DrawLongitudeArc{\t*sin(62)}{50+2}{75}
 \DrawLongitudeArc{(\t-4)*sin(62)}{50+2}{75}
  \DrawLatitudeArc{50+2}{(\t+50-4)*sin(62)}{\t*sin(62)}
 }
 \foreach \t in {0,-50,...,-300} {
   \DrawLatitudeArc{50}{(\t+50-4)*sin(37)}{\t*sin(37)}
 \DrawLongitudeArc{\t*sin(37)}{25+2}{50}
  \DrawLongitudeArc{(\t-4)*sin(37)}{25+2}{50}
   \DrawLatitudeArc{25+2}{(\t+50-4)*sin(37)}{\t*sin(37)}
 }
 \DrawLatitudeArc{50}{(-300-4)*sin(37)}{-330*sin(37)}
 \foreach \t in {0,-50,...,-450} {
    \DrawLatitudeArc{25}{(\t+50-4)*sin(23)}{\t*sin(23)}
 \DrawLongitudeArc{\t*sin(23)}{00+2}{25}
 \DrawLongitudeArc{(\t-4)*sin(23)}{00+2}{25}
 \DrawLatitudeArc{00+2}{(\t+50-4)*sin(23)}{\t*sin(23)}
 }
     \DrawLatitudeArc{25}{(-450-4)*sin(23)}{-500*sin(23)}

\fill[black] (0,3.75) circle (.075cm);

\fill[black] (1.72,3.08) circle (.075cm);
\fill[black] (.76,2.73) circle (.075cm);
\fill[black] (-.66,2.73) circle (.075cm);
\fill[black] (-1.73,3.04) circle (.075cm);

\fill[black] (2.25,1.5) circle (.075cm);
\fill[black] (.8,1.2) circle (.075cm);
\fill[black] (-.85,1.22) circle (.075cm);
\fill[black] (-2.27,1.5) circle (.075cm);
\fill[black] (-3.09,1.97) circle (.075cm);
\fill[black] (3.09,1.97) circle (.075cm);

\fill[black] (2.57,-.15) circle (.075cm);
\fill[black] (1.43,-.37) circle (.075cm);
\fill[black] (.155,-.48) circle (.075cm);
\fill[black] (-1.17,-.41) circle (.075cm);
\fill[black] (-2.35,-.2) circle (.075cm);
\fill[black] (-3.26,0.1) circle (.075cm);
\fill[black] (-3.79,.55) circle (.075cm);
\fill[black] (3.37,.18) circle (.075cm);
\fill[black] (3.85,.57) circle (.075cm);
\end{tikzpicture}
\caption{Patch decomposition of the northern half of the unit sphere.}\label{fig:blub}
\end{figure}

Given a vector $k\in \north$ we define the sets of indices\footnote{Unlike \cite{BNPSS}, where the condition $\hat{k}\cdot\hat{\omega}_\alpha \geq N^{-\delta}$ was used, here we use $k\cdot\hat{\omega}_\alpha \geq N^{-\delta}$. While in the present paper the difference is not important since $|k|$ is bounded, the latter choice is the natural one in \cite[Lemma~5.3]{BNPSS2}, where it means that $c^*_\alpha(k)$ can be bounded by the gapped number operator for $\alpha \in \Ik$ (for all $k$).} (with the same parameter $0 < \delta< 1/6$ we already mentioned in the previous paragraph)
\begin{equation}   \label{eq:cut--off} 
 \begin{split}
\mathcal{I}_{k} &:= \mathcal{I}_{k}^{+} \cup \mathcal{I}_{k}^{-}\;, \\
\mathcal{I}_{k}^{+} &:= \Big\{ \alpha \in \{1, 2, \ldots, M\} \mid k\cdot \hat \omega_{\alpha} \geq N^{-\delta} \Big\}\;, \\
\mathcal{I}_{k}^{-} &:= \Big\{ \alpha \in \{1, 2, \ldots, M\} \mid k\cdot \hat \omega_{\alpha} \leq -N^{-\delta} \Big\}\;.  
 \end{split}  
\end{equation}
That is, (for $N$ large enough) the set $\mathcal{I}_{k}$ only takes into account the labels of the patches which are away from the equator of the Fermi ball, defining as `north' the direction of $k$. Given a patch $B_{\alpha}$ with $\alpha \in \mathcal{I}^{+}_{k}$, we define the pair creation operators
\begin{equation}
b^*_{\alpha}(k) := \frac{1}{m_{\alpha}(k)} \sum_{\substack{ p: p\in B_{\textnormal{F}}^{\textnormal{c}} \cap B_{\alpha} \\ p - k \in B_{\textnormal{F}} \cap B_{\alpha}}} a^{*}_{p} a^{*}_{p-k}\;,\qquad m_{\alpha}(k) :=  \sum_{\substack{ p: p\in B_{\textnormal{F}}^{\textnormal{c}} \cap B_{\alpha} \\ p - k \in B_{\textnormal{F}} \cap B_{\alpha}}} 1\;.
\end{equation}
Also, for $\alpha \in \mathcal{I}_{k}$ we set
\begin{equation} \label{eq:c_alpha}
c^{*}_{\alpha}(k) := \Big\{  \begin{array}{cl} b^{*}_{\alpha}(k) & \quad \textnormal{for }\alpha \in \mathcal{I}^{+}_{k}\;, \\ b^{*}_{\alpha}(-k) & \quad \textnormal{for }\alpha \in \mathcal{I}_{k}^{-}\;. \end{array}
\end{equation}
Let $n_{\alpha}(k) := m_{\alpha}(k)$ for $\alpha \in \mathcal{I}^{+}_{k}$ and $n_{\alpha}(k) := m_{\alpha}(-k)$ for $k\in \mathcal{I}^{-}_{k}$. According to the counting argument from \cite[Proposition 3.1]{BNPSS}, under the assumption $M \gg N^{2\delta}$ we have
\begin{equation}\label{eq:doku}
n_{\alpha}(k)^{2} = \frac{4\pi k_{\textnormal{F}}^{2}}{M} |k\cdot \hat \omega_{\alpha}| \Big(1 + O( \sqrt{M} N^{-\frac{1}{3} + \delta} )\Big) \gg 1\;.
\end{equation}
Note that in \cite{BNPSS} it was assumed that $M \gg N^{1/3}$, which was used only for the linearization of the kinetic energy in expectation value. This condition has been relaxed in \cite{BNPSS2} and in the present paper by linearizing not the operator $\Hbb_0$ in expectation values but only its commutator with the $c^*_\alpha(k)$.

The operators $c_\alpha(k)$ are approximately bosonic\footnote{Unlike the exact bosonization in one--dimensional fermionic systems \cite{ML65} or in spin systems \cite{CG12,CGS15,Ben17,NS19}, the bosonization used here is an approximation, however with rigorous control on the error. Our definition of the pair operators has some similarity to the particle--number conserving operators creating an excitation of a Bose--Einstein condensate introduced in \cite{Gir62,KB62} and used in \cite{LS02,LS04,Sei11,BS19,BBCS18,BBCS19a,BBCS19b,BBCS20}.}   annihilation operators, namely they satisfy the approximate canonical commutation relations  
\[
[c_\alpha(k), c_\beta(\ell)]=0\;, \quad [c_\alpha(k), c_\beta^*(\ell)] \simeq \delta_{k,\ell} \delta_{\alpha,\beta}\;. 
\]
We refer to Lemma \ref{lem:approxCCR} for precise estimates. 

\paragraph{Effective approximately bosonic Hamiltonian.}
Neglecting the corridors between patches and the equatorial region where $k \cdot \hat{\omega}_\alpha \in [0,N^{-\delta})$, we obtain the approximate decomposition
\begin{align}
 b^*(k) \simeq \sum_{\alpha \in \Ikp} n_\alpha(k) b^*_\alpha(k)\;.
\end{align}
With this decomposition we can write $Q_\B$ as an expression that is quadratic in the $c$-- and $c^*$--operators.

We can now construct a quadratic bosonic approximation for $\Hbb_0$. In fact, the reason for decomposing into patches lies in the fact that the $c^*_\alpha(k)$ operators create approximate eigenmodes of the kinetic energy; namely, for $\alpha \in \Ikp$ (and likewise for $\alpha \in \Ikm$) we find
\begin{align}
 [\Hbb_0,c^*_\alpha(k)] & = \Big[ \sum_{i \in \Zbb^3} e(i) a^*_i a_i, \frac{1}{n_\alpha(k)} \sum_{\substack{p: p \in B_\F^c \cap B_\alpha\\p-k \in B_\F \cap B_\alpha}} a^*_p a^*_{p-k} \Big] \label{eq:Hlinearization}\\
 & = \frac{1}{n_\alpha(k)} \sum_{\substack{p: p \in B_\F^c \cap B_\alpha\\p-k \in B_\F \cap B_\alpha}} \left( e(p) + e(p-k) \right) a^*_p a^*_{p-k} \simeq 2\hbar \kappaf \lvert k\cdot \hat{\omega}_\alpha \rvert c^*_\alpha(k)
\end{align}
by linearizing around the point $\omega_\alpha$: $e(p) + e(p-k) = \hbar^2 (2p-k)\cdot k \simeq \hbar^2 (2 k_\F \hat{\omega}_\alpha)\cdot k$. This is the same commutator as would be obtained approximately for $\Hbb_0$ replaced by 
\begin{equation} \label{eq:Dbb}
\Dbb_{\textnormal{B}} := 2\hbar \kappaf \sum_{k\in \north} \sum_{\alpha \in \mathcal{I}_{k}} | k\cdot \hat \omega_{\alpha} | c^{*}_{\alpha}(k) c_{\alpha}(k)\;.
\end{equation}
This approximation will be rigorously justified when the operators act on states obtained from the vacuum by adding bosonic excitations, see \cref{lem:boskin} and \cref{prp:almbos}. Hence, at least on this class of states, we expect that the correlation Hamiltonian can be approximated by
\begin{equation}\label{eq:effbosH}
\Hcal_\textnormal{corr}= R^* \Hcal_N R - E^{\textnormal{pw}}_{N} \simeq \sum_{k\in \north} 2\hbar \kappaf \lvert k\rvert h_{\textnormal{eff}}(k)
\end{equation}
with
\begin{align}\label{eq:heff}
h_\textnormal{eff}(k)  := \!\sum_{\alpha,\beta \in \Ik} \Big[ \big( \D(k) + \W(k) \big)_{\alpha,\beta} c_\alpha^*(k) c_\beta(k) +  \frac{1}{2} \Wt(k)_{\alpha,\beta} \big( c^*_\alpha(k) c^*_\beta(k) + \hc \big)  \Big]
\end{align}
where $\D(k)$, $\W(k)$, and $\Wt(k)$ are  real symmetric matrices of size $\lvert\Ik\rvert \times \lvert\Ik\rvert$ with elements 
\begin{equation}
\label{eq:blocks1}\begin{split}
\D (k)_{\alpha,\beta} &: =  \delta_{\alpha,\beta} \lvert \hat{k} \cdot \hat{\omega}_\alpha\rvert\;, \qquad \hat{k} := \frac{k}{\lvert k\rvert}\;, \qquad \forall \alpha,\beta \in \Ik\;, \\
\W(k)_{\alpha,\beta}  & := \frac{\hat V(k)}{2\hbar \kappaf N \lvert k\rvert} \times \left\{ \begin{array}{cl} n_\alpha(k) n_\beta(k) & \text{ if } \alpha,\beta \in \Ikp \text{ or } \alpha,\beta \in \Ikm \\
 0 & \text{ otherwise}\,,\end{array} \right. \\
\Wt(k)_{\alpha,\beta} & := \frac{\hat V(k)}{2\hbar \kappaf N \lvert k\rvert} \times \left\{ \begin{array}{cl} 
0  & \text{ if } \alpha,\beta \in \Ikp \text{ or } \alpha,\beta \in \Ikm 
\\ n_\alpha(k) n_\beta(k)  & \text{otherwise}\,.
\end{array} \right.
\end{split}
\end{equation}
Note that for all $k \in \Zbb^3$ and all $\alpha, \beta \in \Ik$ we have\footnote{We use the letter $C$ generically for positive constants that may change from line to line.}
\begin{equation} \label{eq:DWWtbound}
 \lvert D(k)_{\alpha,\beta}\rvert \leq \delta_{\alpha,\beta}\;, \qquad \lvert \W(k)_{\alpha,\beta} \rvert \leq \frac{C}{M} \hat{V}(k)\;, \qquad \lvert \Wt(k)_{\alpha,\beta} \rvert \leq \frac{C}{M} \hat{V}(k)\;.
\end{equation}
These bounds immediately imply that the Hilbert--Schmidt norms satisfy
\begin{equation}\label{eq:DWWbounds}
\norm{D(k)}\HS \leq \sqrt{M}\;, \qquad \norm{\widetilde{W}(k)}\HS \leq C\;, \qquad \norm{W(k)}\HS \leq C\;.
\end{equation}
\subsection{Bogoliubov Transformation} 
If $c_\alpha^*(k)$ were exactly bosonic creation operators, then the quadratic Hamiltonian $h_\textnormal{eff}(k)$ could be diagonalized by a \emph{Bogoliubov transformation} \cite[Appendix A.1]{BNPSS}
\begin{equation} \label{eq:Bog-tran}
T_\textnormal{B}(k)  :=   \exp\Big(\frac{1}{2} \sum_{\alpha, \beta \in \mathcal{I}_{k}} K(k)_{\alpha, \beta} c^{*}_{\alpha}(k) c^{*}_{\beta}(k) - \hc\Big)
\end{equation}
where
\begin{equation} 
\label{eq:Kk}
K(k) := \log \lvert S_1 (k)^\intercal \rvert =\frac{1}{2} \log  \Big( S_1(k) S_1(k)^\intercal \Big)
\end{equation}
and
 \begin{align*}
 & S_{1}(k)  := (\D(k)+\W(k)-\Wt(k))^{\! 1/2} E(k)^{-1/2}\;,  \\
& E(k)\!:=\!\!\Big[\! \big(\D(k)+\W(k)-\Wt(k)\big)^{1/2}\! (\D(k)+\W(k)+\Wt(k)) \big(\D(k)+\W(k)-\Wt(k)\big)^{\! 1/2} \Big]^{\! 1/2}\!\!. 
\end{align*}
Then
\[
T^*_\B c_{\gamma}(l) T_\B = \sum_{\alpha\in \mathcal{I}_{l}} \cosh(   K(l) )_{\alpha, \gamma} c_{\alpha}(l) + \sum_{\alpha \in \mathcal{I}_{l}} \sinh (  K(l) )_{\alpha, \gamma} c^{*}_{\alpha}(l) \;.
\]
With this choice of $K(k)$, the ``off--diagonal'' terms in the Hamiltonian (of the form $c^* c^*$ and $c c$) are cancelled by conjugation with the unitary $T_\B$ (see the proof of \cite[Lemma~10.1]{BNPSS2}), so that
\begin{equation} \label{eq:exc-ope-def}
T_\textnormal{B}(k)^*  h_\textnormal{eff}(k) T_\textnormal{B}(k) \simeq \frac{1}{2} \tr (E(k) - \D(k)-\W(k)) + \sum_{\alpha, \beta \in \mathcal{I}_{k}} \mathfrak{K}(k)_{\alpha, \beta} c^{*}_{\alpha}(k) c_{\beta}(k) \;.
\end{equation}
The $\lvert \Ik\rvert\times \lvert \Ik\rvert$--matrix $\mathfrak{K}(k)$ is found to be
\begin{align}
\mathfrak{K}(k) &= \cosh(K(k)) (D(k) + W(k)) \cosh(K(k)) + \sinh(K(k)) (D(k) + W(k)) \sinh(K(k)) \nonumber\\
& \quad + \cosh(K(k)) \widetilde{W}(k) \sinh(K(k)) + \sinh(K(k)) \widetilde{W}(k) \cosh(K(k))\;. \label{eq:curlyKexplicit}
\end{align}
With the orthogonal matrix $O(k)$ defined by the polar decomposition $S_1(k) = O(k) \lvert S_1(k)\rvert$ and $S_2(k) := (S_1(k)^\intercal)^{-1}$ one finds $\cosh(K(k)) = \frac{1}{2}(S_1(k)+S_2(k))O(k)^\intercal$ and $\sinh(K(k)) = \frac{1}{2}(S_1(k)-S_2(k))O(k)^\intercal$. By direct computation this leads to
\begin{align}   \label{eq:curlyK}
\mathfrak{K}(k) &= O(k)^\intercal E(k) O(k) \;.
\end{align}
Thus $h_\textnormal{eff}(k)$ can be understood as the approximately bosonic second quantization of the operator $\mathfrak{K}(k)$ on the one--boson space $\ell^2(\Ik) \simeq \Cbb^{\lvert \Ik\rvert}$. If the effective Hamiltonians at different momenta $k$ were independent, we could simply sum over $k\in \north$ and find that the excitation spectrum consists of sums of eigenvalues of $2\hbar \kappaf \lvert k\rvert E(k)$; see \cite{Ben19} for a discussion of the spectrum. 

In our rigorous application, $c_\alpha^*(k)$ are only approximately bosonic creation operators; moreover, $c_\alpha(k)$ and $c_\alpha(\ell)^*$ do not commute exactly for $k \ne \ell$. Nevertheless, we can still define the unitary transformation 
\begin{equation}    \label{eq:B}
T:= e^{B}\;, \quad B: = \sum_{k\in \north}   \frac{1}{2} \sum_{\alpha, \beta \in \mathcal{I}_{k}} K(k)_{\alpha, \beta} c^{*}_{\alpha}(k) c^{*}_{\beta}(k) - \hc
\end{equation}
and show that it is approximately (see Lemma \ref{lem:Bog}) a bosonic Bogoliubov transformation,  
\begin{equation}
T^{*} c_{\gamma}(l) T  \simeq \sum_{\alpha\in \mathcal{I}_{l}} \cosh(   K(l) )_{\alpha, \gamma} c_{\alpha}(l) + \sum_{\alpha \in \mathcal{I}_{l}} \sinh (  K(l) )_{\alpha, \gamma} c^{*}_{\alpha}(l) \;.
\end{equation}
Consequently, up to error terms that are small on states with few excitations, %
\begin{equation} \label{eq:corr-diag}
T^{*} \Hcal_\textnormal{corr}  T  \simeq  \widetilde E_N^{\textnormal{RPA}} + \sum_{k\in \north} 2\hbar \kappaf \lvert k\rvert \sum_{\alpha, \beta \in \mathcal{I}_{k}} \mathfrak{K}(k)_{\alpha, \beta} c^{*}_{\alpha}(k) c_{\beta}(k)
\end{equation} 
where
\begin{align}  \label{eq:GSE-RPA}
\widetilde E_N^{\textnormal{RPA}}=  \sum_{k\in \north} \hbar \kappaf \lvert k\rvert \tr (E(k) - \D(k)-\W(k))\;. 
\end{align}
The trace can be written out as a sum over $\alpha \in \Ik$. If this sum is seen as a Riemann sum for a surface integral over the sphere, we get (see \cite[Eq.~(5.15)]{BNPSS} for the details) the bound 
\begin{align} \label{eq:RPAdiffs}
\left\lvert \widetilde{E}_{N}^\textnormal{RPA} - {E}_{N}^\textnormal{RPA}\right\rvert \leq C\hbar \left( N^{-\frac{\delta}{2}} + M^{-\frac{1}{4}} N^{\frac{\delta}{2}} + M^{\frac{1}{4}}N^{-\frac{1}{6} + \frac{\delta}{2}} \right)\;, 
\end{align}
where 
\begin{equation}\label{eq:RPA-exp}
{E}^\textnormal{RPA}_N := \hbar \kappaf \sum_{k\in \mathbb{Z}^{3}} |k| \Big( \frac{1}{\pi}\int_{0}^{\infty} \log \Big[ 1 + 2\pi \kappaf \hat V(k) ( 1- \lambda \arctan(\lambda^{-1})) \Big]\di\lambda - \frac{\pi}{2}\kappaf \hat V(k) \Big)\;.
\end{equation} 
 
The key approximation \eqref{eq:corr-diag} has been justified in \cite{BNPSS,BNPSS2,CHN21,BPSS21} for the \emph{expectation value} in a low energy state, at least when  $\hat V$ is non--negative, compactly supported with $\| \hat V \|_{\ell^{1}}$ small enough (but independent of $N$). Optimizing the choice of the parameters $M$ and $\delta$, we obtained the rigorous expansion of the ground state energy $E_N$ of the Hamiltonian $H_N$ in \eqref{eq:HN}:
\begin{equation} \label{eq:rpa_energy}
E_N = E^\textnormal{pw}_N + E_N^{\textnormal{RPA}} + \mathcal{O}(\hbar^{1+\frac{1}{16}}) \;.
\end{equation} 
In the present paper, we justify the approximation \eqref{eq:corr-diag} in \emph{norm} on a class of special states (see \cref{prp:nonbos} and \cref{lem:diag}). From that we obtain a norm approximation for the dynamics \eqref{eq:intro-Sch} for initial data describing pair excitations over an approximate ground state.

\subsection{Main Result: Norm Approximation}
We shall discuss the evolution of states that describe $m$ particle--hole excitations around the Fermi ball. Let $R$ be the particle--hole transformation in \eqref{eq:phtrafo} and $T$ the Bogoliubov transformation in \eqref{eq:Bog-tran}.  We consider the Schr\"odinger equation \eqref{eq:intro-Sch} with the initial state
\begin{equation} \label{eq:xi-def-ini}
\psi = R T \xi  \in L^2_a(\Rbb^{3N})\;,\qquad  \xi = \frac{1}{Z_{m}} c^{*}(\varphi_{1}) \cdots c^{*}(\varphi_{m}) \Omega\;,
\end{equation}
where
\begin{equation}
c^{*}(\varphi_{i}) = \sum_{k \in \north} \sum_{\alpha \in \mathcal{I}_{k}} c^{*}_{\alpha}(k) (\varphi_{i}(k))_{\alpha}
\end{equation}
with $c^{*}_{\alpha}(k)$ being defined in \eqref{eq:c_alpha}, and normalized one--boson wave functions 
\begin{equation}\label{eq:139}
\varphi_{1}, \ldots, \varphi_{m} \in \bigoplus_{k\in \north} \ell^2(\mathcal{I}_k) \;, \quad 
\norm{ \varphi_{i} }^2 := \sum_{k\in \north} \sum_{\alpha \in \mathcal{I}_{k}} \lvert (\varphi_{i}(k))_{\alpha}\rvert^2 = 1\;.
\end{equation}
Note that we do not require orthogonality of the functions $\varphi_i$: since they describe approximately bosonic excitations, they may even all occupy the same one--particle function $\varphi_1$.
The normalization constant $Z_{m}$ (estimated in \cref{prp:norm}) is chosen such that $\| \xi \| = 1$.  
We define the family of time--dependent states
\begin{equation} \label{eq:def-xi-t}
\xi_{t} := \frac{1}{Z_{m}} c^{*}(\varphi_{1;t}) \cdots c^{*}(\varphi_{m;t}) \Omega\;,\qquad t\in \Rbb\;,
\end{equation}
where, with $\mathfrak{K}(k)$ the excitation operator defined in \eqref{eq:corr-diag}, 
\begin{equation}\label{eq:evphi}
\varphi_{m;t} := e^{-i H_{\textnormal{B}} t/\hbar} \varphi_{m}\;,\qquad H_{\textnormal{B}} := \bigoplus_{k\in \north} 2 \hbar \kappaf |k| \mathfrak{K}(k)\;.
\end{equation}
The state $\xi_{t}$ can be viewed as an approximate $m$--particle bosonic state, where every $\varphi_{i}$ evolves according to the one--particle Hamiltonian $H_{\textnormal{B}}$. In general $\xi_t$ is not normalized, but its norm is close to $1$ uniformly in time; this is proven in \cref{lem:norm-xi-s}.

The next theorem is our main result. It provides a norm approximation for the $N$--body evolution of $\psi=R T \xi$ in terms of the explicit states $RT\xi_{t}$ when $N\to \infty$.
\begin{thm}[Norm approximation]\label{thm:1}
Assume that $\hat{V}: \Zbb^3 \to \Rbb$ is compactly supported, non--negative, and satisfies $\hat{V}(k)=\hat{V}(-k)$ for all $k\in \Zbb^3$. Let $k_\F >0$ sufficiently large, $N := \lvert\{ k\in \Zbb^3: \lvert k\rvert \leq k_\F \}\rvert$, and $\hbar := \kappaf k_\F^{-1}$ with $\kappaf = \left(3/4\pi\right)^{\frac{1}{3}}$. Take the number of patches $M:=N^{4\delta}$ with $\delta :=2/45$, the cut--off parameter used in \cref{eq:cut--off}. 

Let $\xi$ be as in \eqref{eq:xi-def-ini} and let $\xi_t$ be as in \eqref{eq:def-xi-t}. Then there exists a constant $C_{m,V} > 0$ depending only on $m$ and $V$ such that for any $t \in \Rbb$ we have
\begin{align}  \label{eq:main-1-weaker-bb}
\norm{ e^{-i\mathcal{H}_{N} t/\hbar} R T \xi - e^{-i (E^{\textnormal{pw}}_{N} + {E}_N^\textnormal{RPA})t/\hbar} R T \xi_{t} } \le  C_{m,V} \hbar^{\frac{1}{15} }  |t|\;. 
\end{align}
\end{thm}

\begin{rems}
\ 
\begin{enumerate}
\item The vector $\psi = RT\xi$ is an $N$--particle state. This is easily verified using $R^*\Ncal R = N + \sum_{p \in B_\F^c} a^*_p a_p - \sum_{h \in B_\F} a^*_h a_h$, and the fact that $\sum_{p \in B_\F^c} a^*_p a_p - \sum_{h \in B_\F} a^*_h a_h$ commutes with all particle--hole pair creation operators $c^*(\varphi_i)$ and with $T$. 

\item We can allow initial data in which the number $m$ of pair excitations grows slowly with $N\to \infty$, as long as we assume $m^3(2m-1)!! \ll N^{\delta}$ (required by \cref{eq:ZB2} to control the normalization constant $Z_m$)\footnote{We use the notation $n!!$ for the semifactorial, i.\,e., $n!! = n(n-2)(n-4)\cdots4\cdot 2$ for even $n \in \Nbb$ and $n!! = n(n-2)(n-4)\cdots3\cdot 1$ for odd $n \in \Nbb$.}. In fact, we have $ C_{m,V} = C_V (m+1)^2\sqrt{(2m-1)!!}$, where $C_V$ depends only on $V$. 
For example, we can take $m \ll \log N / \log \log N$; then by Stirling's formula $C_{m,V}$ grows slower than $N^\epsilon$ for any  $\epsilon>0$.

\item We may also consider initial data constructed with a non--optimal choice of $M=M(N)$ and of $\delta$. In this case the rate of convergence will differ from the $\hbar^{\frac{1}{15}}$ given here, being given instead by 
\begin{align*}
&C  (m+1)^2\sqrt{(2m-1)!!} \times \\
&\times \Big[  \Big( N^{-\frac{\delta}{2}}  + M^{-\frac{1}{2}} + M^{\frac{3}{2}}N^{-\frac{1}{3} + \delta} + M^{\frac{1}{4}} N^{-\frac{1}{6}}  \Big)|t| + M^{-\frac{1}{4}} N^{\frac{\delta}{2}} + M^{\frac{1}{4}}N^{-\frac{1}{6} + \frac{\delta}{2}}   \Big] 
\end{align*}
which is the sum of the error estimates \cref{eq:endofproof} and \cref{eq:lasterr}.

\item The construction of initial data through the patch decomposition may seem quite special; however, the collective pair excitations provide an observable contribution to the excitation spectrum of the many--body system (see \cite{Ben19}).

This procedure has a further big advantage: it provides us with a highly non--trivial tool for the construction of approximate eigenstates (in the sense of being approximately stationary under the many--body evolution) by taking the $\varphi_m(k)$ as eigenvectors of the matrix $\Kfrak(k)$.

\item \label{rem:v} Recently, in \cite{CHN21}, the stationary properties of the same system have been investigated using a different method which does not rely on collective operators. It is unclear to us whether this new technique can be used to study the dynamics. In particular, the non--collective pair operators in \cite{CHN21} behave bosonic only in an ``averaged" sense (unlike our collective pairs), which makes it significantly harder to formulate an effective dynamics. 

\item Let us finally comment on our assumptions on the interaction potential. To apply techniques that have been introduced in \cite{BNPSS2}, we require $\hat{V}$ to have compact support (but, in contrast to  \cite{BNPSS2}, we do not assume the potential to be small). Recently, the results of \cite{BNPSS2} have been extended in \cite{BPSS21,CHN21} to a larger class of potentials, assuming only $\hat{V} (k) \geq 0$ and $\sum |k| \hat{V} (k) < \infty$. Following the ideas of \cite{BPSS21} it would certainly be possible to extend \cref{thm:1} to the same class of interactions. However, to keep the presentation as transparent as possible, we prefer to restrict our analysis here to potentials with compact support (moreover, extension to more general potential would lead to a deterioration of the error estimate \cref{eq:main-1-weaker-bb}).    
\end{enumerate}
\end{rems}
At first sight, our result looks similar to the norm approximations obtained for bosonic systems in, e.\,g.,  \cite{GM13,LNS15,BCS17,NN17,BNNS19}.
However, there is an important difference: for bosonic mean--field systems, 
an effective quadratic Hamiltonian arises by quasifree reduction (see, e.\,g., \cite{BSS18}); instead for our fermionic norm approximation, the formal effective Hamiltonian \cref{eq:effbosH} is quartic in fermionic operators. Only through bosonization can we approximate it as a quasifree, and thus solvable, effective theory.  
\paragraph{Organization of the paper.} 
The rest of the paper is devoted to the proof of \cref{thm:1}.
 In Section \ref{sec:pre} we recall estimates from \cite{BPSS21,BNPSS,BNPSS2}. In Section \ref{sec:refo}, we explain how Theorem \ref{thm:1} follows if we justify the approximation \eqref{eq:corr-diag} in Fock space norm. This is then undertaken in the following sections: In  Section \ref{sec:red-bos} we reduce the correlation Hamiltonian $\mathcal{H}_\textnormal{corr}$ to its bosonizable terms. In Section \ref{sec:lin} we prove that the fermionic kinetic operator $\mathbb{H}_{0}$ can be replaced by a  bosonized one. In Section \ref{sec:diag}, we diagonalize the resulting approximate bosonic Hamiltonian by a Bogoliubov transformation.
 In Section \ref{sec:final} we combine all estimates to prove \cref{thm:1}.

\section{Approximate Bosonization: Key Estimates}\label{sec:pre}

In this section we recall important estimates from \cite{BPSS21,BNPSS,BNPSS2}.
\begin{lem}[{\cite[Eq.~(4.10)]{BPSS21}}] There is a $C > 0$ (independent of $N$) such that for any $k\in \mathbb{Z}^3$ we have
\begin{equation} \label{eq:HPR20}
\sum_{p\in B_\textnormal{F}^c \cap (B_\textnormal{F}+k)} \frac{1}{e(p)+e(p-k)} \le C N\;. 
\end{equation}
\end{lem}

In the following we summarize the properties of the operators $c_{\alpha}(k)$ and  $c_{\alpha}(k)^{*}$. The next lemma shows that, on states with few excitations, they behave as bosonic operators.
\begin{lem}[Approximate CCR, {\cite[Lemma~4.1]{BNPSS} and \cite[Eq.~(5.3)]{BNPSS2}}]\label{lem:approxCCR} Let $k, l \in \north$, $\alpha \in \mathcal{I}_{k}$, and $\beta \in \mathcal{I}_{l}$. Then
\begin{equation}\label{eq:approxCCR}
[ c_{\alpha}(k), c_{\beta}(l) ] = [ c^{*}_{\alpha}(k), c^{*}_{\beta}(l) ] = 0\;,\qquad [ c_{\alpha}(k), c^{*}_{\beta}(l) ] =: \delta_{\alpha, \beta} ( \delta_{k,l} + \mathcal{E}_{\alpha}(k,l) )\;.
\end{equation}
The operator $\mathcal{E}_{\alpha}(k,l)$ commutes with the fermionic number operator $\Ncal$ and satisfies
\begin{equation}\label{eq:ptwisebdE-11}
\sum_{\beta \in \mathcal{I}_{k} \cap \mathcal{I}_{l}} \mathcal{E}_{\beta}(k,l)^* \mathcal{E}_{\beta}(k,l) = \sum_{\beta \in \mathcal{I}_{k} \cap \mathcal{I}_{l}} |\mathcal{E}_{\beta}(k,l)|^{2} \leq C(M N^{-\frac{2}{3} + \delta} \Ncal)^{2} 
\end{equation}
and for all $\zeta \in \mathcal{F}$ also
\begin{equation}
\sum_{\alpha \in \mathcal{I}_{k}\cap \mathcal{I}_{l}} \| \mathcal{E}_{\alpha}(k,l)\zeta \| \leq CM^{\frac{3}{2}} N^{-\frac{2}{3} + \delta} \| \Ncal\zeta \|\;.
\end{equation}
Furthermore $\mathcal{E}_{\alpha}(k,l) = \mathcal{E}_{\alpha}(l,k)^{*}$ for all $\alpha \in \Ik$ and all $k,l\in \north$.
\end{lem}

 In the proof of the lower bound \cite{BNPSS2} an important role was played by the gapped number operator $
\Ncal_\delta := \sum_{i \in \Zbb^3\colon e(i) \geq \frac{1}{4} N^{-\frac{1}{3} -\delta}} a^*_i a_i$: since modes of very low energy are excluded, this operator can be more efficiently bounded by the kinetic energy $\Hbb_0$.
In the present paper we consider only explicitly constructed states for which we have strong control on $\Ncal$, so that the use of $\Ncal_\delta$ is not necessary. In the following lemmas we replaced all appearances of $\Ncal_\delta$ using $\Ncal_\delta \leq \Ncal$.

The patch operators $c_{\alpha}(k)$, $c_{\alpha}(k)^{*}$ satisfy similar bounds as true bosonic operators:
\begin{lem}[Pair operators bounds, {\cite[Lemma~5.3]{BNPSS2}}]\label{lem:bosbd}
For all $k\in \north$ we have
\begin{equation}\label{eq:c*c}
\sum_{\alpha \in \mathcal{I}_{k}} c^{*}_{\alpha}(k) c_{\alpha}(k) \leq \Ncal\;.
\end{equation}
Consequently, for all $\zeta \in \mathcal{F}$, we have
\begin{equation}\label{eq:sumc}
\sum_{\alpha \in \mathcal{I}_{k}} \| c_{\alpha}(k) \zeta \| \leq M^{\frac{1}{2}} \| \Ncal^{\frac{1}{2}} \zeta \|\;,\qquad \sum_{\alpha \in \mathcal{I}_{k}} \| c^{*}_{\alpha}(k) \zeta \| \leq M^{\frac{1}{2}} \| (\Ncal + M)^{\frac{1}{2}}\zeta \|\;.
\end{equation}
Moreover, for $f\in \ell^{2}(\mathcal{I}_{k})$,
\begin{equation} \label{eq:c*-f}
\Big\| \sum_{\alpha \in \mathcal{I}_{k}} f_{\alpha} c^{*}_{\alpha}(k) \zeta \Big\| \leq \| f \|_{\ell^{2}} \| (\Ncal + 1)^{\frac{1}{2}} \zeta \|\;.
\end{equation}
Since $c^*_\alpha(k)$ is (a sum over) a product of two fermionic creation operators we have
\begin{equation}\label{eq:Nccomm}[\Ncal,c^*_\alpha(k)] = 2 c^*_\alpha(k) \qquad \text{and} \qquad g(\Ncal) c^*_\alpha(k) = c^*_\alpha(k) g(\Ncal+2)
 \end{equation}
 for any measurable function $g: \Rbb \to \Rbb$.
\end{lem}

 The next lemma allows us to understand the action of the kinetic energy operator $\mathbb{H}_{0}$ in terms of an approximately bosonized operator $\mathbb{D}_{\textnormal{B}}$ defined by \cref{eq:Dbb}. The patch decomposition is necessary for the linearization of the dispersion relation that justifies the identity $[ \mathbb{H}_{0}, c^{*}_{\alpha}(k) ] \simeq 2\hbar \kappaf | k\cdot \hat \omega_{\alpha} | c^{*}_{\alpha}(k)$.
\begin{lem}[Bosonization of the kinetic energy]\label{lem:boskin}
For all $k\in \north$ and $\alpha \in \mathcal{I}_{k}$ we have
\begin{equation}\label{eq:linearization}\begin{split}
[ \mathbb{H}_{0}, c^{*}_{\alpha}(k) ] &= 2\hbar \kappaf | k\cdot \hat \omega_{\alpha} | c^{*}_{\alpha}(k) + \hbar \mathfrak{E}^{\textnormal{lin}}_{\alpha}(k)^{*}\;,\\
\protect{[} \mathbb{D}_{\textnormal{B}}, c^{*}_{\alpha}(k) \protect{]} &= 2\hbar \kappaf | k\cdot \hat \omega_{\alpha} | c^{*}_{\alpha}(k) + \hbar \mathfrak{E}^{\textnormal{B}}_{\alpha}(k)^{*}\;,
\end{split}\end{equation}
where for all $\zeta \in \mathcal{F}$ the error terms are bounded by
\begin{equation}\label{eq:sumfrakE}
\begin{split}
\sum_{\alpha\in \mathcal{I}_{k}} \| \mathfrak{E}^{\textnormal{lin}}_{\alpha}(k) \zeta \| &\leq C\| \Ncal^{\frac{1}{2}}\zeta \|\;,\\
\sum_{\alpha \in \mathcal{I}_{k}} \| \mathfrak{E}^{\textnormal{B}}_{\alpha}(k)\zeta \| &\leq CM^{\frac{3}{2}} N^{-\frac{2}{3} + \delta} \| (\Ncal + 1)^{\frac{3}{2}} \zeta \|\;.
\end{split}
\end{equation}
The operator $\mathfrak{E}^{\textnormal{lin}}_{\alpha}(k)$ commutes with $c_{\beta}(k)$ for all $\alpha, \beta =1,\ldots,M$. Instead $\mathfrak{E}^{\textnormal{B}}_{\alpha}(k)$ commutes with $c_{\beta}(k)$ assuming that $\alpha\neq \beta$. Finally, for $f\in \ell^{2}(\mathcal{I}_{k})$, we have
\begin{equation}\label{eq:sumEE}
\begin{split}
\Big\| \sum_{\alpha\in \mathcal{I}_{k}} f_{\alpha} \mathfrak{E}^{\textnormal{lin}}_{\alpha}(k)^{*} \zeta \Big\| &\leq C \|f\|_{\ell^{2}} M^{-\frac{1}{2}} \| (\Ncal + 1)^{\frac{1}{2}} \zeta \| \;,\\
\Big\| \sum_{\alpha \in \mathcal{I}_{k}} f_{\alpha} \mathfrak{E}^{\textnormal{B}}_{\alpha}(k)^{*} \zeta \Big\| &\leq C\|f\|_{\ell^{2}} M^{\frac{3}{2}} N^{-\frac{2}{3} + \delta} \| (\Ncal + 1)^{\frac{3}{2}} \zeta \|\;.
\end{split}
\end{equation}
The operators $\mathfrak{E}^{\textnormal{lin}}_{\alpha}(k)$ and $\mathfrak{E}^{\textnormal{B}}_{\alpha}(k)$ annihilate two fermions, i.\,e.,
\begin{align*}
[\Ncal,\mathfrak{E}^{\textnormal{lin}}_{\alpha}(k) ]  = -2 \mathfrak{E}^{\textnormal{lin}}_{\alpha}(k)\;,  \qquad  [\Ncal,\mathfrak{E}^{\textnormal{B}}_{\alpha}(k)] = -2 \mathfrak{E}^{\textnormal{B}}_{\alpha}(k) \;. 
\end{align*}
\end{lem}
\begin{proof}
 The lemma collects the results of \cite[Lemma 8.2]{BNPSS2} and of \cite[Eq.~(8.4)--Eq.~(8.6)]{BNPSS2}. The second bound in (\ref{eq:sumEE}) easily follows from the explicit expression \cite[Eq.~(8.5)]{BNPSS2} ($\chi(\alpha \in \Il)$ denotes an indicator function)
\begin{align} \label{eq:e-lin-d}
  \Efrak^{\linb}_\alpha(k) := 2 \kappaf \sum_{l \in \north}  \lvert k \cdot \hat{\omega}_\alpha\rvert   \Ecal^*_\alpha(l,k) c_\alpha(l) \chi(\alpha \in \Il)\;, 
  \end{align}
  and from the bounds for $\mathcal{E}_{\alpha}(l,k)$ and $c^{*}_{\alpha}(l)$ that were given in \cite[Lemma 5.2]{BNPSS2} and \cite[Lemma 5.3]{BNPSS2}, respectively.
\end{proof}

We turn to the approximate Bogoliubov transformation $T$. The Bogoliubov kernel $K(k)$, defined in \cref{eq:Kk}, is controlled by the following bound. The bound is stronger than the one given in \cite[Lemma 4.5]{BNPSS} and weaker than the bound given in \cite[Lemma~6.1]{BNPSS2}. Compared to the latter, it has the advantage of not requiring a small--potential assumption.
\begin{lem}[Bogoliubov kernel] \label{lem:K}
For all $k\in \north$ we have
\begin{equation}\label{eq:Kresult}
\lvert K(k)_{\alpha,\beta}\rvert \leq C \frac{\hat{V}(k)}{M} \quad \forall \alpha,\beta \in \Ik\;.
\end{equation}  
Furthermore
\begin{equation}\label{eq:Kbound}
 \norm{K(k)}\HS \leq C\;.
\end{equation}
\end{lem}
\begin{proof}
 We drop the $k$--dependence from the notation for all matrices.
We write
\begin{equation} \label{eq:defv}
g:=\frac{1}{2}\kappaf \hat V(k)\;, \qquad u_\alpha:= \sqrt{\lvert \hat{k} \cdot \hat{\omega}_\alpha\rvert}\;, \qquad v_\alpha:= \frac{\hbar}{\kappaf \sqrt{\lvert k\rvert}} n_\alpha(k)\;, \qquad \forall \alpha \in\Ik\;. 
\end{equation}
Note that the coefficient $v_\alpha$ from the interaction part can be controlled by the coefficient $u_\alpha$ from the kinetic energy: according to \cref{eq:doku} we have
\[v_\alpha = \sqrt{\frac{4\pi}{M}} u_\alpha \left( 1 + \Ocal(\sqrt{M}N^{-\frac{1}{3}+\delta}) \right) \leq \frac{C}{\sqrt{M}} u_\alpha\;.\]
Due to the reflection symmetry of the construction of patches we have
\[
B_{\alpha+M/2} = - B_\alpha\;, \quad \omega_{\alpha+M/2} = - \omega_\alpha \qquad \forall \alpha\in\{1,2,\ldots, M/2\} 
\]
so that the matrices in \cref{eq:blocks1} can be written in block form as 
\[
\D = \begin{pmatrix} d & 0 \\ 0 & d \end{pmatrix}, \quad \label{eq:Wdef}\W = \begin{pmatrix} b & 0 \\ 0 & b \end{pmatrix}\;, \quad \Wt = \begin{pmatrix} 0& b\\ b & 0\end{pmatrix}\;,
\]
where
\[d=\diag(u_\alpha^2, \alpha = 1,\ldots,\ik) \quad \textnormal{and} \quad b=g\lvert v\rangle \langle v\rvert\;,\]
the latter indicating the rank--one operator with $v = (v_1, \cdots , v_{\ik})$. (The indices are delimited by $I := \lvert \Ikp \rvert = \lvert \Ikm \rvert$.)
 
 According to \cref{eq:Kk} we have, with $L := S_1 S_1^\intercal - \id$,
 \begin{equation}
  K = \frac{1}{2} \log(\id+L) \quad \Leftrightarrow \quad e^{2K} = \id + L \;.
 \end{equation}
 Using the orthogonal matrix (as in \cite{GS13} and \cite[Eq.~(6.8)]{BNPSS2})
 \[U := \frac{1}{\sqrt{2}} \begin{pmatrix} \id & \id \\ \id & -\id \end{pmatrix}\]
 we can block--diagonalize
 \begin{equation}
  U L U^\intercal = \begin{pmatrix} L_1 & 0 \\ 0 & L_2 \end{pmatrix}\;,
  \end{equation}
  where
  \begin{equation}\begin{split}
 L_1 & = d^{1/2}\left[ d^{1/2} (d+2b) d^{1/2} \right]^{-1/2} d^{1/2} - \id \;,\\
 L_2 & = (d+2b)^{1/2} \left[ (d+2b)^{1/2}  d (d+2b)^{1/2} \right]^{-1/2} (d+2b)^{1/2} - \id\;.
 \end{split}
 \end{equation}
We can also write
\[U K U^\intercal = \begin{pmatrix} K_1 & 0 \\ 0 & K_2 \end{pmatrix}\;, \quad \textnormal{where} \quad e^{2 K_i} = L_i + 1 \quad \forall i \in \{1,2\}\;.\]
Note that
\[
 L_2 \geq 0 \geq L_1 \quad \Leftrightarrow \quad K_2 \geq 0 \geq K_1 \;.
\]
In \cite[Eq.~(6.21) and (6.30)~et~seqq.]{BNPSS2} we showed that
\begin{equation}
 \lvert (L_i)_{\alpha,\beta} \rvert \leq C \frac{\hat{V}(k)}{M} \min \left\{ \frac{u_\alpha}{u_\beta}, \frac{u_\beta}{u_\alpha} \right\} \leq C \frac{\hat{V}(k)}{M} \;.
\end{equation}
This was proven in \cite{BNPSS2} without using the smallness assumption on the potential. (The smallness assumption was only used to transfer this bound for $L_i$ to a bound for $K_i$ via a series expansion of the logarithm.) Here we avoid the use of a series expansion. This way, we can still obtain the important factor $M^{-1}$ (but not $\min \left\{ \frac{u_\alpha}{u_\beta}, \frac{u_\beta}{u_\alpha} \right\}$).

\paragraph{Bound for $K_2$.} Using the operator inequality
\[L_2 = e^{2 K_2} - 1 \geq 2 K_2\]
and the Cauchy--Schwarz inequality we can bound
\begin{equation}
 \lvert (K_2)_{\alpha,\beta} \rvert \leq \sqrt{(K_2)_{\alpha,\alpha}} \sqrt{(K_2)_{\beta,\beta}} \leq \frac{1}{2} \sqrt{(L_2)_{\alpha,\alpha}} \sqrt{(L_2)_{\beta,\beta}} \leq C \frac{\hat{V}(k)}{M}\;.
\end{equation}

\paragraph{Bound for $K_1$.} This is slightly more difficult because $K_1 \leq 0$. We write
\[ \tilde{L}_1 := (L_1 + 1)^{-1} - 1 = e^{-2K_1} - 1 \geq -2 K_1 \;.\]
We claim that
\begin{equation}\label{eq:claim}
\lvert (\tilde{L}_1)_{\alpha,\alpha} \rvert \leq C \frac{\hat{V}(k)}{M} \quad \forall \alpha\in \{1,2,\ldots, I\}\;.
\end{equation}
Given this claim, we can conclude by the same Cauchy--Schwarz estimate as above that
\begin{equation}
 \lvert (-K_1)_{\alpha,\beta} \rvert \leq \sqrt{(-K_1)_{\alpha,\alpha}} \sqrt{(-K_1)_{\beta,\beta}} \leq \frac{1}{2} \sqrt{(\tilde{L}_1)_{\alpha,\alpha}} \sqrt{(\tilde{L}_1)_{\beta,\beta}} \leq C \frac{\hat{V}(k)}{M}\;.
\end{equation}
To show \cref{eq:claim}, we write
\begin{align}
 \tilde{L}_1 & = (L_1 + 1)^{-1} - 1 = d^{-1/2}\left[ d^{1/2} (d+2b) d^{1/2} \right]^{1/2} d^{-1/2} - 1 \nonumber \\
 & = d^{-1/2}\left( \left[ d^{1/2} (d+2b) d^{1/2} \right]^{1/2} - d \right) d^{-1/2}=: d^{-1/2} A d^{-1/2}\;. \label{eq:220}
\end{align}
Recall the integral identity for the matrix square root $X^{1/2} = -\frac{2}{\pi} \int_0^\infty  \lambda^2 (X + \lambda^2)^{-1} \di\lambda$ and the formula for the inverse of a matrix with rank--one perturbation $(X+\lvert x\rangle \langle y \rvert)^{-1} = X^{-1} - X^{-1}\lvert x\rangle \langle y \rvert X^{-1} (1+\langle y,X^{-1} x\rangle)^{-1}$. Writing
\[
d^{1/2} (d+2b) d^{1/2} = d^2 + 2g \lvert \tilde{v}\rangle\langle \tilde{v}\rvert \qquad \textnormal{with} \quad \tilde{v} := d^{1/2} v
\]
we find
\begin{align*}
 \lvert A_{\alpha,\alpha}\rvert & =  \left\lvert \frac{2}{\pi}\int_0^\infty \lambda^2 \left( \frac{1}{d^2+\lambda^2} - \frac{1}{d^2 + \lambda^2 + 2g \lvert \tilde{v}\rangle \langle \tilde{v} \rvert}\right)_{\alpha,\alpha} \di\lambda \right\rvert \\
 & \leq \frac{4g}{\pi} \int_0^\infty \frac{\lambda^2}{1+2g \langle \tilde{v}, (d^2 + \lambda^2)^{-1} \tilde{v}\rangle} \left\lvert \left( \frac{1}{d^2 + \lambda^2} \lvert \tilde{v}\rangle \langle \tilde{v} \rvert \frac{1}{d^2 + \lambda^2} \right)_{\alpha,\alpha} \right\rvert \di\lambda
\end{align*}
Note that $\langle \tilde{v}, (d^2 + \lambda^2)^{-1} \tilde{v}\rangle \geq 0$ can be dropped from the denominator for an upper bound, so that using the explicit form of the matrix elements according to \cref{eq:defv} we get
 \begin{align}
  \lvert A_{\alpha,\alpha}\rvert & \leq \frac{Cg}{M} \int_0^\infty \lambda^2  \frac{u_\alpha^4}{(u_\alpha^4 + \lambda^2)^2} \di\lambda = \frac{C g}{M} u_\alpha^2 \;.
 \end{align} 
 According to \cref{eq:220}, since $\left(d^{-1/2}\right)_{\alpha,\beta} = \delta_{\alpha,\beta} u_\alpha^{-1}$ is diagonal, this implies
the claimed bound \cref{eq:claim}. 

\paragraph{Conclusion.} In summary, we proved that for both $i\in \{1,2\}$ we have
\[
 \lvert (K_i)_{\alpha,\beta} \rvert \leq C \frac{\hat{V}(k)}{M} \quad \forall \alpha,\beta \in \{1,2,\ldots, I\}\;.
\]
Recalling
\[
K = U^\intercal \begin{pmatrix} K_1 & 0 \\ 0 & K_2 \end{pmatrix} U = \frac{1}{2}\begin{pmatrix} K_1 + K_2 & K_1 - K_2 \\ K_1 - K_2 & K_1 + K_2 \end{pmatrix}
\]
we arrive at \cref{eq:Kresult}.

\medskip

The bound for the Hilbert--Schmidt norm follows trivially from the first bound.
\end{proof}

The next lemma shows that the number operator does not increase significantly under conjugation with the operator $T_\lambda := \exp(\lambda B)$, where the operator $B$ is defined in \cref{eq:B}.
\begin{lem}[Stability of the number operator, {\cite[Proposition 4.6]{BNPSS}}]\label{lem:stabN} There exists a $C > 0$ such that for all $n \in \Nbb$ and all $\lambda \in [0,1]$ we have
\begin{equation}
T_{\lambda}^* ( \Ncal + 1 )^{n} T_{\lambda} \leq e^{Cn}  (\Ncal + 1)^{n}\;. 
\end{equation}
\end{lem}
The following lemma shows that $T_{\lambda}$ acts as an approximate Bogoliubov transformation on the pair operators. %
\begin{lem}[Approximate Bogoliubov transformation] \label{lem:Bog}\label{prp:NkE} Let $l\in \north$ and $\gamma \in \mathcal{I}_{l}$. Then there exists $C >0$ such that for all $\lambda \in [0,1]$ and all $m\in \frac{1}{2}\mathbb{N}$ we have
\begin{equation}
T^{*}_{\lambda} c_{\gamma}(l) T_{\lambda} = \sum_{\alpha\in \mathcal{I}_{l}} \cosh( \lambda K(l) )_{\alpha, \gamma} c_{\alpha}(l) + \sum_{\alpha \in \mathcal{I}_{l}} \sinh ( \lambda K(l) )_{\alpha, \gamma} c^{*}_{\alpha}(l) + \mathfrak{E}_{\gamma}(\lambda,l)\;,
\end{equation} 
with the bound
\begin{equation}
\sum_{\gamma\in \mathcal{I}_{l}} \| \Ncal^{m} \mathfrak{E}_{\gamma}^{*}(\lambda,l) \zeta \| \leq CM^{\frac{3}{2}} N^{-\frac{2}{3} + \delta} \|  (\Ncal + 1)^{m+\frac{3}{2}} \zeta \| \quad \forall \zeta\in \fock \;.
\end{equation}
The same estimate holds with $\mathfrak{E}_{\gamma}^{*}(\lambda,l)$ replaced by $\mathfrak{E}_{\gamma}(\lambda,l)$.
\end{lem}
\begin{proof} The error bound here is a generalization of \cite[Lemma 7.1]{BNPSS2}, which only considered $m=0$. See also \cite[Proposition 4.4]{BNPSS} for an earlier related result. Recall from \cite[Eq.~(7.9)]{BNPSS2} that by an iterated Duhamel expansion of the conjugation $T^*_\lambda( \cdot )T_\lambda$ up to arbitrary order $n_0 \in \Nbb$, the error term is 
\begin{align*}
\mathfrak{E}_{\gamma}^{*}(\lambda,l) & = \sum_{n=0}^{n_{0}-1} \int_{0}^{\lambda} \di\tau\, \frac{(\lambda - \tau)^{n}}{n!} \sum_{k\in \north} \sum_{\alpha \in \mathcal{I}_{l} \cap \mathcal{I}_{k}} \sum_{\beta\in \mathcal{I}_{k}} ( K(l)^{n} )_{\gamma, \alpha} K(k)_{\alpha, \beta}\\
&\hspace{10em} \times T^{*}_{\tau} \frac{1}{2} \Big( \mathcal{E}_{\alpha}(k,l) c^{*}_{\beta}(k) + c^{*}_{\beta}(k) \mathcal{E}_{\alpha}(k,l) \Big)^{\natural} T_{\tau}\\
&\quad+ \int_{0}^{\lambda} \di t\, \frac{(\lambda - \tau)^{n_{0}-1}}{(n_{0} - 1)!} \sum_{\alpha\in \mathcal{I}_{l}} ( K(l)^{n_{0}} )_{\gamma,\alpha} T^{*}_{\tau} c^{\natural}_{\alpha}(l) T_{\tau} \\
& \quad - \sum_{\alpha \in \mathcal{I}_{l}} \sum_{n = n_{0}}^{\infty} \frac{\lambda^{n} (K(l)^{n})_{\gamma,\alpha}}{n!} c^{\natural}_{\alpha}(l)\;.
\end{align*} 
The symbol $\natural$ as superscript is used as an abbreviation to indicate either the operator itself or its adjoint (namely $A^\natural$ may be either $A^*$ or $A$), where the choice between the two options does not play a role for the further estimates. \Cref{lem:K} together with the bound $\| K(k) \|_\textnormal{op} \leq C$ implies the estimate $| (K(l)^n)_{\gamma,\alpha}| \leq C^n M^{-1}$, valid without any smallness condition on the potential and for all $n \in \Nbb \setminus \{0\}$. Then, proceeding as in \cite[Eq.~(7.8)--(7.9)]{BNPSS2} we get:
\begin{align*}
& \sum_{\gamma \in \mathcal{I}_{l}} \| \Ncal^{m} \mathfrak{E}_{\gamma}^{*}(\lambda,l) \zeta\| \\
& \leq \sum_{n=0}^{n_{0}-1} \frac{C^{n}}{n! M} \int_{0}^{\lambda} \di\tau\, \sum_{k\in \north} \sum_{\substack{\alpha \in \mathcal{I}_{l} \cap \mathcal{I}_{k}\\\beta\in \mathcal{I}_{k}}} \Big( \| \Ncal^{m} T^{*}_{\tau} \mathcal{E}_{\alpha}(k,l) c^{*}_{\beta} T_{\tau} \zeta \| + \| \Ncal^{m} T^{*}_{\tau} c^{*}_{\beta}(k) \mathcal{E}_{\alpha}(k,l) T_{\tau} \zeta \| \\
&\hspace{13em} + \| \Ncal^{m} T^{*}_{\tau}c_{\beta}(k) \mathcal{E}^{*}_{\alpha}(k,l) T_{\tau} \zeta \| + \| \Ncal^{m} T^{*}_{\tau} \mathcal{E}^{*}_{\alpha}(k,l) c_{\beta}(k) T_{\tau} \zeta \| \Big)\\
&\quad + \frac{C^{n_{0}}}{(n_{0} - 1)!} \int_{0}^{\lambda} \di\tau\, \sum_{\alpha\in \mathcal{I}_{l}} \Big( \| \Ncal^{m} T^{*}_{\tau} c_{\alpha}(l) T_{\tau} \zeta \| + \| \Ncal^{m} T^{*}_{\tau} c^{*}_{\alpha}(l) T_{\tau} \zeta \| \Big)\\
&\quad + \sum_{n=n_{0}}^{\infty} \frac{C^{n}}{n!} \sum_{\alpha \in \mathcal{I}_{l}} \Big( \| \Ncal^{m} c_{\alpha}(l) \zeta \| + \| \Ncal^{m} c^{*}_{\alpha}(l) \zeta \| \Big)\;.
\end{align*}
By the stability of $\Ncal^{m}$ under conjugation with $T_{\tau}$ (see \cref{lem:stabN}), we get
\begin{displaymath}
\norm{ \Ncal^{m} T^{*}_{\tau} \zeta } \leq C_{m} \norm{ (\Ncal + 1)^{m} \zeta }\;,\qquad \forall \zeta \in \mathcal{F}\;.
\end{displaymath}
For the next step, recall that $\mathcal{E}_{\alpha}(k,l)$ commutes with the number operator, and since $c_{\alpha}(k)$ and $c_{\alpha}^{*}(k)$ annihilate and create a pair of fermions, respectively, we have $c_\alpha(k) \Ncal = (\Ncal-2)c_\alpha(k)$. If we use $\Ncal^{m} \leq (\Ncal + 1)^{m}$ before, we thus get
\begin{align*}
&\sum_{\gamma \in \mathcal{I}_{l}} \| \Ncal^{m} \mathfrak{E}_{\gamma}^{*}(\lambda,l) \zeta\|  \\
& \leq \sum_{n=0}^{n_{0}-1}\! \frac{C^{n}}{n! M}\! \int_{0}^{\lambda}\!\! \di\tau\!\! \sum_{k\in \north}\! \sum_{\substack{\alpha \in \mathcal{I}_{l} \cap \mathcal{I}_{k}\\\beta\in \mathcal{I}_{k}}}\!\!\! \Big( \| \mathcal{E}_{\alpha}(k,l) c^{*}_{\beta} (\Ncal + 3)^{m}T_{\tau} \zeta \| + \| c^{*}_{\beta}(k) \mathcal{E}_{\alpha}(k,l) (\Ncal + 3)^{m}T_{\tau} \zeta \| \\
&\hspace{14em}  + \| c_{\beta}(k)\mathcal{E}^{*}_{\alpha}(k,l) \Ncal^{m} T_{\tau} \zeta \| + \|  \mathcal{E}^{*}_{\alpha}(k,l) c_{\beta}(k) \Ncal^{m} T_{\tau} \zeta \|\Big)\\
&\quad + \frac{C^{n_{0}}}{(n_{0} - 1)!} \int_{0}^{\lambda} \di\tau\, \sum_{\alpha\in \mathcal{I}_{l}} \Big( \| c_{\alpha}(l) \Ncal^{m} T_{\tau} \zeta \| + \| c^{*}_{\alpha}(l) (\Ncal + 3)^{m} T_{\tau} \zeta \| \Big)\\
&\quad + \sum_{n=n_{0}}^{\infty} \frac{C^{n}}{n!} \sum_{\alpha \in \mathcal{I}_{l}} \Big( \| c_{\alpha}(l) \Ncal^{m} \zeta \| + \| c^{*}_{\alpha}(l) (\Ncal + 3)^{m} \zeta \| \Big)\;.
\end{align*} 
From this point, we proceed exactly as in the proof of \cite[Lemma~7.1]{BNPSS2}. Let us sketch the proof for completeness.  Using \cref{lem:bosbd} 
we can bound the operators in the last two lines; then taking $n_0\to \infty$ we obtain 
\begin{align*}
&\sum_{\gamma \in \mathcal{I}_{l}} \| \Ncal^{m} \mathfrak{E}_{\gamma}^{*}(\lambda,l) \zeta\|  \\
& \leq CM^{-1} \int_{0}^{\lambda}\!\! \di\tau\!\! \sum_{k\in \north}\! \sum_{\substack{\alpha \in \mathcal{I}_{l} \cap \mathcal{I}_{k}\\\beta\in \mathcal{I}_{k}}}\!\!\! \Big( \| \mathcal{E}_{\alpha}(k,l) c^{*}_{\beta} (\Ncal + 3)^{m}T_{\tau} \zeta \| + \| c^{*}_{\beta}(k) \mathcal{E}_{\alpha}(k,l) (\Ncal + 3)^{m}T_{\tau} \zeta \| \\
&\hspace{14em}  + \| c_{\beta}(k)\mathcal{E}^{*}_{\alpha}(k,l) \Ncal^{m} T_{\tau} \zeta \| + \|  \mathcal{E}^{*}_{\alpha}(k,l) c_{\beta}(k) \Ncal^{m} T_{\tau} \zeta \|\Big)\\
&=: I_1+ I_2 + I_3 +I_4
\end{align*} 
which is similar to \cite[Eq.~(7.10)]{BNPSS2}. By proceeding similarly to \cite[Eq. (7.11)]{BNPSS2} (replacing $T_\tau \Psi$ by $(\Ncal+3)^mT_\tau \Psi$ and using $\Ncal_\delta \le \Ncal$, accordingly), we have
\begin{align*}
I_1 &\le \sup_{\tau\in [-1,1]} CM^{-1}  \sum_{k\in \north}\! \sum_{\substack{\alpha \in \mathcal{I}_{l} \cap \mathcal{I}_{k}\\\beta\in \mathcal{I}_{k}}}\!\!\!   \| \mathcal{E}_{\alpha}(k,l) c^{*}_{\beta} (\Ncal + 3)^{m}T_{\tau} \zeta \| \\
&\le  \sup_{\tau\in [-1,1]} CM N^{-\frac 2 3 +\delta}    \|  (\Ncal +M)^{1/2} (\Ncal +1) (\Ncal + 3)^{m} T_{\tau} \zeta \| \\
&\le \sup_{\tau\in [-1,1]} CM^{\frac 3 2} N^{-\frac 2 3 +\delta}    \|  (\Ncal +1)^{m+\frac 3 2}  T_{\tau} \zeta \| \\
&\le CM^{\frac 3 2} N^{-\frac 2 3 +\delta}    \|  (\Ncal +1)^{m+\frac 3 2}  \zeta \|. 
 \end{align*} 
In the last estimate we have used the stability of $\Ncal^{m}$ under conjugation with $T_{\tau}$ (recall \cref{lem:stabN}). Next, by proceeding similarly to the argument leading to \cite[Eq.~(7.14)]{BNPSS2} (replacing $T_\tau \Psi$ by $(\Ncal+3)^mT_\tau \Psi$ and using $\Ncal_\delta \le \Ncal$ again), we find that
\begin{align*}
I_2 &\le   \sup_{\tau\in [-1,1]} CM N^{-\frac 2 3 +\delta}    \|  (\Ncal +M)^{1/2} (\Ncal +1) (\Ncal + 3)^{m} T_{\tau} \zeta \| \\
&\le \sup_{\tau\in [-1,1]} CM^{\frac 3 2} N^{-\frac 2 3 +\delta}    \|  (\Ncal +1)^{m+\frac 3 2}  T_{\tau} \zeta \| \\
&\le CM^{\frac 3 2} N^{-\frac 2 3 +\delta}    \|  (\Ncal +1)^{m+\frac 3 2}  \zeta \|. 
 \end{align*} 
The error terms $I_3$ and $I_4$ can be treated by the same way. The proof of \cref{prp:NkE} is complete. 
\end{proof}

\section{Reformulation of Norm Approximation} \label{sec:refo}

In this section, we reduce Theorem \ref{thm:1} to an appropriate version of the key approximation \eqref{eq:corr-diag}. We define
\[
\mathfrak{H}_{\textnormal{exc}}:= \sum_{k\in \north} 2\hbar \kappaf \lvert k\rvert \sum_{\alpha, \beta \in \mathcal{I}_{k}} \mathfrak{K}(k)_{\alpha, \beta} c^{*}_{\alpha}(k) c_{\beta}(k)\;.
\]

\begin{lem}[Reduction to norm approximation of the Hamiltonian] Let $\xi_t$ be defined as in \eqref{eq:def-xi-t}. Then for all $t\ge 0$ 
\begin{align} \label{eq:norm-red-0}
&\norm{ e^{-i\mathcal{H}_{N} t/\hbar} R T \xi - e^{-i (E^{\textnormal{pw}}_{N} + \widetilde E_N^\textnormal{RPA})t/\hbar} R T \xi_{t} }  \nonumber \\
&\le  \frac{1}{\hbar} \int_{0}^{t} \big\| (T^*\mathcal{H}_\textnormal{corr} T - \widetilde E_N^\textnormal{RPA} - \mathfrak{H}_{\textnormal{exc}})  \xi_{s}  \big\| \di s + C \frac{m^2 \sqrt{(2m-1)!!}}{Z_{m}} M^{\frac{3}{2}} N^{-\frac{2}{3} + \delta} t\;. 
\end{align}
\end{lem}

\begin{proof} We start by writing, using the unitarity of $e^{-i\mathcal{H}_{N} t/\hbar}$, $R$, and $T$, 
\begin{align} \label{eq:norm-red-1}
&\norm{ e^{-i\mathcal{H}_{N} t/\hbar} R T \xi - e^{-i (E^{\textnormal{pw}}_{N} + \widetilde E_N^\textnormal{RPA})t/\hbar} R T \xi_{t} } \nonumber\\
& = \Big\|  \int_{0}^{t} \di s\, \frac{\di}{\di s} \Big( e^{i (\mathcal{H}_{N} - E^{\textnormal{pw}}_{N} - \widetilde E_N^\textnormal{RPA})s/\hbar} R T \xi_{s} \Big) \Big\|\nonumber\\
&\le \frac{1}{\hbar} \int_{0}^{t}\di s\, \big\| (\mathcal{H}_{N} - E^{\textnormal{pw}}_{N} - \widetilde E_N^\textnormal{RPA}) R T \xi_{s} - R T (i\hbar \partial_s \xi_{s} )  \big\| \nonumber \\
&= \frac{1}{\hbar} \int_{0}^{t}\di s\, \big\| T^* R^* (\mathcal{H}_{N} - E^{\textnormal{pw}}_{N} - \widetilde E_N^\textnormal{RPA}) RT \xi_{s} -  i\hbar \partial_s \xi_{s}   \big\|\nonumber\\
&= \frac{1}{\hbar} \int_{0}^{t}\di s\, \big\|  \left(T^* \mathcal{H}_\textnormal{corr} T  - \widetilde E_N^\textnormal{RPA} \right) \xi_{s} -  i\hbar \partial_s \xi_{s}  \big\| \;.
\end{align}
 From  \eqref{eq:def-xi-t} we have 
\begin{equation}
i\hbar \partial_s \xi_{s} = \frac{1}{Z_{m}}\sum_{i=1}^{m} c^{*}(\varphi_{1;s}) \cdots c^{*}(\varphi_{i-1;s}) c^{*}(H_{\textnormal{B}} \varphi_{i;s}) c^{*}(\varphi_{i+1;s}) \cdots c^{*}(\varphi_{m;s}) \Omega\;,
\end{equation}
where $H_{\textnormal{B}} = \bigoplus_{k\in \north} 2 \hbar\kappaf |k| \mathfrak{K}(k)$. Using that $\mathfrak{H}_{\textnormal{exc}} \Omega = 0$, we get
\begin{align}
\mathfrak{H}_{\textnormal{exc}} \xi_{s} &= \frac{1}{Z_{m}}\mathfrak{H}_{\textnormal{exc}} c^{*}(\varphi_{1;s}) \cdots c^{*}(\varphi_{m;s}) \Omega = \frac{1}{Z_{m}}[ \mathfrak{H}_{\textnormal{exc}}, c^{*}(\varphi_{1;s}) \cdots c^{*}(\varphi_{m;s})] \Omega \nonumber\\
&= \frac{1}{Z_{m}} \sum_{i=1}^{m} c^{*}(\varphi_{1;s}) \cdots c^{*}(\varphi_{i-1;s}) [ \mathfrak{H}_{\textnormal{exc}}, c^{*}(\varphi_{i;s}) ] c^{*}(\varphi_{i+1;s}) \cdots c^{*}(\varphi_{m;s}) \Omega\;.\nonumber
\end{align}
The commutator equals
\begin{equation}\label{eq:Hc}
[ \mathfrak{H}_{\textnormal{exc}}, c^{*}(\varphi_{i;s}) ]  = \sum_{k\in \north} 2\hbar \kappaf \lvert k\rvert \sum_{\alpha, \beta \in \mathcal{I}_{k}}  \mathfrak{K}(k)_{\alpha,\beta} c^{*}_{\alpha}(k) [c_{\beta}(k), c^{*}(\varphi_{i;s})]\;,
\end{equation}
where according to \cref{eq:approxCCR} we have
\begin{align}
[c_{\beta}(k), c^{*}(\varphi_{i;s})] &= \sum_{l\in \north} \sum_{\gamma \in \mathcal{I}_{l}} (\varphi_{i;s}(l))_{\gamma} [ c_{\beta}(k), c^{*}_{\gamma}(l) ] \nonumber\\
&= (\varphi_{i;s}(k))_{\beta} + \sum_{l\in \north} (\varphi_{i;s}(l))_{\beta} \mathcal{E}_{\beta}(k,l) \chi(\beta \in \mathcal{I}_{k} \cap \mathcal{I}_{l})\;.
\end{align}
Therefore
\begin{align}\label{eq:Hcstar}
&[ \mathfrak{H}_{\textnormal{exc}}, c^{*}(\varphi_{i;s}) ] = c^{*}(H_{\textnormal{B}} \varphi_{i;s}) \\
&\quad+ \sum_{k\in \north} 2\hbar \kappaf |k| \sum_{\alpha, \beta \in \mathcal{I}_{k}}  \mathfrak{K}(k)_{\alpha,\beta} c^{*}_{\alpha}(k) \sum_{l\in \north} (\varphi_{i;s}(l))_{\beta} \mathcal{E}_{\beta}(k,l) \chi(\beta \in \mathcal{I}_{k} \cap \mathcal{I}_{l})\;.\nonumber
\end{align}
The first term in \cref{eq:Hcstar} is precisely what we need to reconstruct $i\hbar \partial_s \xi_{s}$. Hence
\begin{align}\label{eq:bosdyn}
 \| \mathfrak{H}_{\textnormal{exc}} \xi_{s} - i \hbar \partial_s \xi_{s} \| & \leq \frac{1}{Z_{m}} \sum_{i = 1}^{m}  \sum_{k\in \north} 2\hbar \kappaf |k| \sum_{\beta \in \mathcal{I}_{k}}  \sum_{l\in \north} |(\varphi_{i;s}(l))_{\beta}| \chi(\beta \in \mathcal{I}_{k} \cap \mathcal{I}_{l})\\
&\ \ \times \| c^{*}(\varphi_{1;s}) \cdots c^{*}(\varphi_{i-1;s}) c^{*}( \mathfrak{K}(k)_{\beta}) \mathcal{E}_{\beta}(k,l) c^{*}(\varphi_{i+1;s}) \cdots c^{*}(\varphi_{m;s}) \Omega\|\,, \nonumber
\end{align}
with $\Kfrak(k)_\beta$ for fixed $\beta$ being the $\Cbb^{\lvert \Ik\rvert}$--vector with elements $\Kfrak(k)_{\alpha,\beta}$. We then estimate
\begin{align}\label{eq:bosdyn2}
&\| c^{*}(\varphi_{1;s}) \cdots c^{*}(\varphi_{i-1;s}) c^{*}( \mathfrak{K}(k)_{\beta}) \mathcal{E}_{\beta}(k,l) c^{*}(\varphi_{i+1;s}) \cdots c^{*}(\varphi_{m;s}) \Omega\|\nonumber\\
&\quad \leq \| \prod_{j=1}^{i-1} (\Ncal + 1 + 2(i - 1 -j))^{1/2} c^{*}( \mathfrak{K}(k)_{\beta}) \mathcal{E}_{\beta}(k,l) c^{*}(\varphi_{i+1;s}) \cdots c^{*}(\varphi_{m;s}) \Omega\| \nonumber\\
&\quad \leq \prod_{j=1}^{i-1} (2 (m - j) + 1)^{\frac{1}{2}} \| c^{*}( \mathfrak{K}(k)_{\beta}) \mathcal{E}_{\beta}(k,l) c^{*}(\varphi_{i+1;s}) \cdots c^{*}(\varphi_{m;s}) \Omega\|\;.
\end{align}
Using the crude bound $\norm{\Kfrak(k)_\beta}_{\ell^2} \leq \norm{\Kfrak(k)}\HS \leq C \sqrt{M}$ (which follows from \cref{eq:curlyKexplicit} with \cref{eq:DWWbounds} and \cref{eq:Kbound}), and recalling that $\Ecal_\beta(k,l)$ commutes with $\Ncal$, we estimate
\begin{align}\label{eq:sumbeta}
&\sum_{\beta \in \mathcal{I}_{k} \cap \mathcal{I}_{l}} |(\varphi_{i;s}(l))_{\beta}| \| c^{*}( \mathfrak{K}(k)_{\beta}) \mathcal{E}_{\beta}(k,l) c^{*}(\varphi_{i+1;s}) \cdots c^{*}(\varphi_{m;s}) \Omega\| \nonumber\\
&\quad \leq C M^{\frac{1}{2}}\sum_{\beta \in \mathcal{I}_{k} \cap \mathcal{I}_{l}} |(\varphi_{i;s}(l))_{\beta}| \| (\Ncal + 1)^{\frac{1}{2}} \mathcal{E}_{\beta}(k,l) c^{*}(\varphi_{i+1;s}) \cdots c^{*}(\varphi_{m;s}) \Omega\| \nonumber\\
&\quad \leq C M^{\frac{1}{2}} \Big(\sum_{\beta \in \mathcal{I}_{k} \cap \mathcal{I}_{l}} \| \mathcal{E}_{\beta}(k,l) (\Ncal + 1)^{\frac{1}{2}} c^{*}(\varphi_{i+1;s}) \cdots c^{*}(\varphi_{m;s}) \Omega\|^{2} \Big)^{\frac{1}{2}}\;,
\end{align}
where in the last step we used the Cauchy--Schwarz inequality. With \eqref{eq:ptwisebdE-11} we get
\begin{align}\label{eq:bosdyn3}
&\sum_{\beta \in \mathcal{I}_{k} \cap \mathcal{I}_{l}} |(\varphi_{i;s}(l))_{\beta}| \| c^{*}( \mathfrak{K}(k)_{\beta}) \mathcal{E}_{\beta}(k,l) c^{*}(\varphi_{i+1;s}) \cdots c^{*}(\varphi_{m;s}) \Omega\| \nonumber\\
&\qquad \leq C M^{\frac{3}{2}} N^{-\frac{2}{3} + \delta} \| \Ncal (\Ncal+1)^{1/2} c^{*}(\varphi_{i+1;s}) \cdots c^{*}(\varphi_{m;s})\Omega \| \nonumber\\
&\qquad \leq C M^{\frac{3}{2}} N^{-\frac{2}{3} + \delta} (2(m - i) + 1)^{\frac{3}{2}} \| c^{*}(\varphi_{i+1;s}) \cdots c^{*}(\varphi_{m;s})\Omega \| \nonumber\\
&\qquad \leq C M^{\frac{3}{2}} N^{-\frac{2}{3} + \delta} (2(m - i) + 1)^{\frac{3}{2}} \prod_{j=0}^{m-i-1} (2j + 1)^{\frac{1}{2}}\;.
\end{align}
Combining \cref{eq:bosdyn}, \cref{eq:bosdyn2}, and \cref{eq:bosdyn3}, we obtain 
\begin{align}\label{eq:e5}
 \| i\hbar \partial_s \xi_s - \mathfrak{H}_{\textnormal{exc}} \xi_s  \| \leq C \frac{m^2\sqrt{(2m-1)!!}}{Z_{m}} \hbar M^{\frac{3}{2}} N^{-\frac{2}{3} + \delta}\;. 
\end{align}
From \eqref{eq:norm-red-1} and \eqref{eq:e5}, we obtain \eqref{eq:norm-red-0} by the triangle inequality.  
\end{proof}

Next, we estimate the constant $Z_m$, defined such that $\xi$ in \cref{eq:xi-def-ini} is normalized, by comparing it to $Z_{\textnormal{B};m}$, which would be its value if the $c^*$--operators were exactly bosonic. 

\begin{lem}[Estimate for $Z_m$]\label{prp:norm}
Let the functions $\varphi_i$ be normalized, as in \cref{eq:139}. Let
\begin{equation}\label{eq:ZB}
Z_{\textnormal{B};m} := \sqrt{\sum_{\pi \in S_{m}} \prod_{i=1}^{m} \langle \varphi_{i}, \varphi_{\pi(i)} \rangle}\;,
\end{equation}
where $S_{m}$ is the set of permutations of $\{1,2,3,\ldots, m\}$. Then $1\le Z^{2}_{\textnormal{B};m} \le m!$ and 
\begin{equation}\label{eq:ZB2}
\Big| Z^{2}_{m} - Z^{2}_{\textnormal{B};m}\Big| \leq C M N^{-\frac{2}{3} + \delta} m^3 (2m-1)!!\;.
\end{equation}
\end{lem}
\begin{rem} The bounds $1\leq Z_{\textnormal{B};m}^{2} \leq m!$ are optimal. The case $Z_{\textnormal{B};m}^{2} = 1$ holds for an orthonormal set $(\varphi_{i})_{i=1}^m$, while $Z_{\textnormal{B};m}^{2} = m!$ holds for $\varphi_{1} = \varphi_{2} = \ldots = \varphi_{m}$.
\end{rem}
\begin{proof}[Proof of \cref{prp:norm}.]
By definition of $Z_{m}$, and using that $c(\varphi_{1})\Omega = 0$, we get
\begin{align}
Z_{m}^{2} &= \langle c^{*}(\varphi_{1}) \cdots c^{*}(\varphi_{m}) \Omega, c^{*}(\varphi_{1}) \cdots c^{*}(\varphi_{m}) \Omega \rangle \nonumber\\
&= \langle c(\varphi_{1}) c^{*}(\varphi_{1}) \cdots c^{*}(\varphi_{m}) \Omega, c^{*}(\varphi_{2}) \cdots c^{*}(\varphi_{m}) \Omega \rangle \nonumber\\
&= \langle [ c(\varphi_{1}), c^{*}(\varphi_{1}) \cdots c^{*}(\varphi_{m}) ] \Omega, c^{*}(\varphi_{2}) \cdots c^{*}(\varphi_{m}) \Omega \rangle\;.
\end{align}
Next we expand
\begin{equation}
[ c(\varphi_{1}), c^{*}(\varphi_{1}) \cdots c^{*}(\varphi_{m}) ] = \sum_{i=1}^{m} c^{*}(\varphi_{1}) \cdots c^{*}(\varphi_{i-1}) [c(\varphi_{1}), c^{*}(\varphi_{i})]  c^{*}(\varphi_{i+1}) \cdots c^{*}(\varphi_{m})\;,
\end{equation}
where
\begin{equation}\label{eq:comm1i}
[c(\varphi_{1}), c^{*}(\varphi_{i})] = \langle \varphi_{1}, \varphi_{i} \rangle + \sum_{k, l\in \north} \sum_{\alpha \in \mathcal{I}_{k} \cap \mathcal{I}_{l}} \overline{(\varphi_{1}(k))_{\alpha}} (\varphi_{i}(l))_{\alpha} \mathcal{E}_{\alpha}(k,l)\;.
\end{equation}
(The inner product here is defined in correspondence to the norm \cref{eq:139} as
$
 \langle \varphi_{1}, \varphi_{i} \rangle := \sum_{k\in \north} \sum_{\alpha \in \mathcal{I}_{k}} \overline{(\varphi_{1}(k))_{\alpha}} (\varphi_{i}(k))_{\alpha}$.)
The first term reproduces the commutation relation of exactly bosonic operators, while the second term is an error term. To estimate it, we use that, for any $\psi, \zeta \in \mathcal{F}$,
\begin{displaymath}
 \sum_{\alpha \in \mathcal{I}_{k} \cap \mathcal{I}_{l}} \lvert \overline{(\varphi_{1}(k))_{\alpha}} \rvert \lvert (\varphi_{i}(l))_{\alpha} \rvert \norm{ \mathcal{E}_{\alpha}(k,l) \psi} \norm{  \zeta } \leq \| \varphi_{1}(k) \|_{\ell^{2}} \| \varphi_{i}(l) \|_{\ell^{2}} \max_{\alpha \in \mathcal{I}_{k}\cap \mathcal{I}_{l}} \| \mathcal{E}_{\alpha}(k,l) \psi \| \|\zeta\|\;.
\end{displaymath}
Thanks to \cref{lem:approxCCR}, we have
\begin{equation}\label{eq:Ealpha}
\| \mathcal{E}_{\alpha}(k,l) \psi \| \leq CM N^{-\frac{2}{3} + \delta} \| \Ncal \psi \|\;.
\end{equation}
Hence
\begin{equation}\label{eq:Zm}
Z_{m}^{2} = \sum_{i=1}^{m} \langle \varphi_{1}, \varphi_{i} \rangle \langle c^{*}(\varphi_{1}) \cdots {c}^{*}(\varphi_{i-1}) {c}^{*}(\varphi_{i+1})\cdots c^{*}(\varphi_{m}) \Omega, c^{*}(\varphi_{2}) \cdots c^{*}(\varphi_{m}) \Omega \rangle + \mathfrak{r}_{1}
\end{equation}
with $\mathfrak{r}_{1}$ the contribution produced by the second term in \cref{eq:comm1i}, i.\,e.,
\begin{align} \label{eq:rfrak1}
\mathfrak{r}_{1} & = \sum_{i=1}^{m}\sum_{k, l\in \north} \sum_{\alpha \in \mathcal{I}_{k} \cap \mathcal{I}_{l}}\overline{(\varphi_{1}(k))_{\alpha}} (\varphi_{i}(l))_{\alpha}\\
&\quad\times \Big\langle c^{*}(\varphi_{1}) \cdots c^{*}(\varphi_{i-1}) \mathcal{E}_{\alpha}(k,l) c^{*}(\varphi_{i+1}) \cdots c^{*}(\varphi_{m}) \Omega, c^{*}(\varphi_{2}) \cdots c^{*}(\varphi_{m}) \Omega\Big\rangle\;.\nonumber
\end{align}
Using the estimate $\| c^{*}(\varphi_{i}) \psi \| \leq \| \varphi_{i} \|_{2} \| (\Ncal + 1)^{\frac{1}{2}} \psi \|$ following from \cref{lem:bosbd}, commuting $(\Ncal+1)^{1/2}$ with $\Ecal_\alpha(k,l)$, and using \cref{eq:Ealpha}, we get
\begin{align}\label{eq:estr0}
\lvert \mathfrak{r}_{1} \rvert &\leq m^{2}C M N^{-\frac{2}{3} + \delta} \norm{ (\Ncal + 1)^{\frac{1}{2}} \cdots (\Ncal + 2m - 1)^{\frac{1}{2}}  \Omega}^{2} \nonumber\\
&\leq  C M N^{-\frac{2}{3} + \delta} m^2(2m-1)!!\;.
\end{align}
We iterate the process $m$ times, starting from \cref{eq:Zm}, until we are left with $Z^{2}_{m} = Z^{2}_{\textnormal{B};m} + \mathfrak{r}_{m}$, with $\mathfrak{r}_{m}$ an error term. The error term $\mathfrak{r}_{m}$ is estimated by $m$ times the bound for $\mathfrak{r}_{1}$ in \cref{eq:estr0}, thus giving the bound \cref{eq:ZB2}.
\end{proof}

The time--dependent state $\xi_{s}$ is not necessarily normalized, but in the next lemma we show that its norm is close to $1$. 

\begin{lem}[Norm of $\xi_s$] \label{lem:norm-xi-s} For all $s\in \Rbb$ we have
\begin{equation}
\big|\| \xi_{s} \|^{2} - 1\big| \leq \frac{1}{Z_m^{2}}  C M N^{-\frac{2}{3} + \delta} m^3 (2m-1)!! \;.
\end{equation}
\end{lem}

\begin{proof}
We have
\begin{equation}
\| \xi_{s} \|^{2} = \frac{1}{Z_{m}^{2}} \langle c^{*}(\varphi_{1;s}) \cdots c^{*}(\varphi_{m;s}) \Omega, c^{*}(\varphi_{1;s}) \cdots c^{*}(\varphi_{m;s}) \Omega \rangle\;.
\end{equation}
Proceeding as in the proof of \cref{prp:norm}, with $\rfrak_{m;s}$ defined in analogy to \cref{eq:rfrak1} using the evolved $\varphi_{j;s}$, we have
\begin{align}
\langle c^{*}(\varphi_{1;s}) \cdots c^{*}(\varphi_{m;s}) \Omega, c^{*}(\varphi_{1;s}) \cdots c^{*}(\varphi_{m;s}) \Omega \rangle = \sum_{\pi \in S_{m}} \prod_{i=1}^{m} \langle \varphi_{i;s}, \varphi_{\pi(i);s} \rangle + \mathfrak{r}_{m;s}\;.
\end{align}
By unitarity of the dynamics generated by $H_{\textnormal{B}}$, the first term is precisely $Z_{\textnormal{B};m}^{2}$. The error term $\mathfrak{r}_{m;s}$ is estimated exactly as $\mathfrak{r}_{m}$ in the proof of \cref{prp:norm}.\end{proof}

\section{Reduction to Bosonizable Correlation Hamiltonian} \label{sec:red-bos}
In the introduction we already mentioned the non--bosonizable terms in the correlation Hamiltonian $\Hcal_\textnormal{corr}$; as computed in \cite[Eq.~(1.17)--(1.21)]{BNPSS2}, their precise form is
\begin{align}\label{eq:Hcorrdeterr}
\Xbb & = - \frac{1}{2N} \sum_{k \in \Zbb^3} \hat{V}(k) \bigg[ \sum_{p \in \BFc \cap (\BF +k)} a^*_p a_p  + \sum_{h \in \BF \cap (\BFc-k)} a^*_{h} a_{h}\bigg]\;,\nonumber\\
\mathcal{E}_{1} &= \frac{1}{2N} \sum_{k\in \north} \hat V(k) \Big[ \mathfrak{D}(k)^{*} \mathfrak{D}(k) +  \mathfrak{D}(-k)^{*} \mathfrak{D}(-k) \Big]\;,\nonumber\\
\mathcal{E}_{2} &= \frac{1}{2N} \sum_{k\in \north} \hat V(k) \Big[ \mathfrak{D}(-k)^{*} b(k) + \mathfrak{D}(k)^{*} b(-k) + \hc \Big]\;.
\end{align}
The operator $\Dfrak(k)^* = \Dfrak(-k)$ creates and annihilates particles that are either both outside or both inside the Fermi ball,
\begin{equation} \label{eq:Dkdef}
\Dfrak(k)^* := \sum_{p  \in \BFc \cap (\BFc+k)} a^*_{p} a_{p-k} - \sum_{h  \in \BF \cap (\BF-k)} a^*_{h} a_{h+k} \;.
\end{equation} 

As a first step to establish the approximation \eqref{eq:corr-diag}, let us reduce the correlation Hamiltonian $\Hcal_\textnormal{corr}$ to only its terms $\Hbb_0 + Q_{\textnormal{B}}^{\mathcal{R}}$, where 
\begin{align}
&Q_{\textnormal{B}}^{\mathcal{R}} = \frac{1}{N} \sum_{k\in \north} \hat V(k) \Big[\sum_{\alpha, \beta \in \mathcal{I}^{+}_{k}} n_{\alpha}(k) n_{\beta}(k) c^{*}_{\alpha}(k) c_{\beta}(k) + \sum_{\alpha,\beta\in \mathcal{I}^{-}_{k}} n_{\alpha}(k) n_{\beta}(k) c^{*}_{\alpha}(k) c_{\beta}(k) \nonumber\\
&\hspace{8em} + \sum_{\substack{\alpha\in \mathcal{I}^{+}_{k}\\ \beta\in \mathcal{I}^{-}_{k}}} n_{\alpha}(k) n_{\beta}(k) c^{*}_{\alpha}(k) c^{*}_{\beta}(k) +  \sum_{\substack{\alpha\in \mathcal{I}^{+}_{k}\\\beta\in \mathcal{I}^{-}_{k}}} n_{\alpha}(k) n_{\beta}(k) c_{\beta}(k) c_{\alpha}(k)\Big]\;.
\end{align}
The error due to the dropped terms is controlled by the following lemma.
\begin{lem}[Non--bosonizable terms]\label{prp:nonbos} There exists $C >0$ such that for all $\zeta \in \mathcal{F}$
\begin{align}\label{eq:nonbos}
 &\big\| (\mathcal{H}_{\textnormal{corr}} - \Hbb_0 - Q_{\textnormal{B}}^{\mathcal{R}}) \zeta \big\| \nonumber\\
 &\leq C \|\hat V\|_{1} \hbar \Big(N^{-\frac{2}{3}} \|\Ncal^{2} \zeta \| + N^{-\frac{1}{3}} \| \Ncal^{\frac{3}{2}}\zeta \| + (N^{-\frac{\delta}{2}} + N^{-\frac{1}{6}} M^{\frac{1}{4}} ) \| (\Ncal + 1) \zeta \| \Big)\;. 
\end{align}
\end{lem}

\begin{proof} Recall the correlation Hamiltonian \cref{eq:Hcorr}. By the triangle inequality
\begin{equation}\label{eq:start}
 \big\| (\mathcal{H}_{\textnormal{corr}} -  \Hbb_0 - Q_{\textnormal{B}}^{\mathcal{R}}) \zeta \big\| \leq \| \mathbb{X} \zeta \| +  \| \mathcal{E}_{1} \zeta \| + \| \mathcal{E}_{2} \zeta \| + \| (Q_{\textnormal{B}} - Q_{\textnormal{B}}^{\mathcal{R}}) \zeta\|\;. 
\end{equation}

Consider the first term in \eqref{eq:start}. We have
\begin{align}
\| \mathbb{X} \zeta \| &\leq 
 \frac{1}{N} \sum_{k \in \mathbb{Z}^{3}} |\hat V(k)| \Big\| \sum_{h\in B_{\textnormal{F}} \cap (B_{\textnormal{F}} + k)} a^{*}_{h} a_{h} \zeta \Big\| + \frac{1}{N} \sum_{k \in \mathbb{Z}^{3}} |\hat V(k)| \Big\| \sum_{p\in B^{c}_{\textnormal{F}} \cap (B_{\textnormal{F}} + k)} a^{*}_{p} a_{p} \zeta \Big\| \nonumber\\
& + \frac{1}{2N} \sum_{k\in \mathbb{Z}^{3}} |\hat V(k)| \Big\| \sum_{h\in B_{\textnormal{F}} \cap ( B_{\textnormal{F}} + k)} a^{*}_{h} a_{h} \zeta \Big\| + \frac{1}{2N} \sum_{k\in \mathbb{Z}^{3}} |\hat V(k)| \Big\| \sum_{p\in B^{\textnormal{c}}_{\textnormal{F}} \cap ( B^{\textnormal{c}}_{\textnormal{F}} +k )} a^{*}_{p}a_{p} \zeta \Big\|\;.\nonumber
\end{align}
These four terms are estimated in exactly the same way. For the first, we have
\begin{align}\label{eq:astara}
\Big\| \sum_{h\in B_{\textnormal{F}} \cap (B_{\textnormal{F}} + k)} a^{*}_{h} a_{h} \zeta \Big\|^{2} &= \Big\langle \zeta, \sum_{h,h' \in B_{\textnormal{F}} \cap (B_{\textnormal{F}} + k)} a^{*}_{h'} a_{h'} a^{*}_{h} a_{h} \zeta \Big\rangle \nonumber\\
&= \Big\langle \zeta, \sum_{h,h' \in B_{\textnormal{F}} \cap (B_{\textnormal{F}} + k)} a^{*}_{h} a^{*}_{h'} a_{h'} a_{h} \zeta \Big\rangle + \sum_{h\in B_{\textnormal{F}} \cap (B_{\textnormal{F}} + k)} \langle \zeta, a^{*}_{h} a_{h} \zeta \rangle\nonumber\\
&\le \langle \zeta, (\Ncal+\Ncal^2)  \zeta \rangle \le 2 \langle \zeta, \Ncal^2  \zeta \rangle\;.
\end{align}
The same estimates hold for the other three contributions to $\mathbb{X}$. Hence
\begin{equation}\label{eq:bdX}
\| \mathbb{X} \zeta \| \leq \frac{C\|\hat V\|_{1}}{N} \| \Ncal \zeta \|\;.
\end{equation}

Consider the second term in \cref{eq:start},
\begin{equation}\label{eq:bdE1}
\norm{ \mathcal{E}_{1}\zeta } \leq \frac{1}{N} \sum_{k\in \Zbb^3} \lvert \hat V(k)\rvert \norm{\mathfrak{D}(-k) \mathfrak{D}(k) \zeta }\;.
\end{equation}
Recalling the definition of $\mathfrak{D}(k)$ in \cref{eq:Dkdef} and proceeding as in \cref{eq:astara}, we easily get
\begin{equation}\label{eq:Dbd}
\| \mathfrak{D}(k) \zeta \| \leq C\| \Ncal \zeta \|\;,
\end{equation}
which implies   
\begin{equation} \label{eq:E1}
\| \mathcal{E}_{1}\zeta \| \leq \frac{C \|\hat V\|_{1}}{N}  \| \Ncal^{2} \zeta \|\;.
\end{equation}

Consider the third term in \cref{eq:start}. We have
\begin{equation}\label{eq:E2d}
\| \mathcal{E}_{2} \zeta \| \leq \frac{1}{N} \sum_{k \in \north} |\hat V(k)| \| \mathfrak{D}(k) b(k)\zeta  \| + \frac{1}{N} \sum_{k\in \north} |\hat V(k)| \|  b^{*}(k) \mathfrak{D}(-k)\zeta  \|\;.
\end{equation}
Using $\| a^{\natural}_{p}\| \leq 1$ we get
\begin{align} \label{eq:bk}
\|b(k) \zeta\| &\le \sum_{p\in B_\textnormal{F}^c \cap (B_\textnormal{F}+k)} \|a_{p-k} a_{p} \zeta\| \leq \sum_{p\in B_\textnormal{F}^c \cap (B_\textnormal{F}+k)} \|a_{p} \zeta\|  \nonumber\\
&\le \sqrt{\sum_{p\in B_\textnormal{F}^c \cap (B_\textnormal{F}+k)} 1}  \sqrt{\sum_{p\in B_\textnormal{F}^c \cap (B_\textnormal{F}+k)} \|a_{p} \zeta\|^2} \le CN^{\frac{1}{3}} \| \Ncal^{\frac{1}{2}} \zeta \|\;, \quad \forall \zeta\in \Fcal\;.
\end{align}
Here we used that $\sum_{p\in B_\textnormal{F}^c \cap (B_\textnormal{F}+k)} 1 \leq CN^{\frac{2}{3}}$ by \cite[Eq.~(2.1)]{BNPSS2}.
Using also \cref{eq:Dbd}, this implies that the first summand of \cref{eq:E2d} is bounded by 
\begin{align*}
\frac{1}{N}\sum_{k\in \north} |\hat V(k)| \| \mathfrak{D}(k) b(k)\zeta\| & \leq \frac{1}{N} \sum_{k\in \north} |\hat V(k)| C \norm{\Ncal b(k) \xi} \\
& = \frac{1}{N} \sum_{k\in \north} |\hat V(k)| C \norm{b(k) (\Ncal - 2) \xi}\\
 & \leq C \|\hat V\|_{1} N^{-\frac{2}{3}} \| \Ncal^{\frac{3}{2}} \zeta \| \;.
\end{align*}
To bound the second term in \cref{eq:E2d}, we compute
 \begin{align} \label{eq:bk-bk*}
  [b(k), b^{*}(k)] &= \sum_{p\in B_{\textnormal{F}}^{\textnormal{c}} \cap (B_{\textnormal{F}} + k)} (1 - a^{*}_{p}a_{p} - a^{*}_{p-k} a_{p-k})
 \le  \sum_{p\in B_{\textnormal{F}}^{\textnormal{c}} \cap (B_{\textnormal{F}} + k)} 1 \le  CN^{\frac{2}{3}}\;. 
 \end{align}
Therefore
\begin{align}
\|  b^{*}(k) \mathfrak{D}(-k)\zeta  \|^2 &\leq CN^{\frac{2}{3}} \| \mathfrak{D}(-k)\zeta  \|^2 + \|  b(k) \mathfrak{D}(-k)\zeta  \|^2 \nonumber\\
&\leq CN^{\frac{2}{3}}\| \Ncal \zeta \|^2 + CN^{\frac{2}{3}} \| \Ncal^{\frac{3}{2}}\zeta \|^2\;,
\end{align}
which implies 
\begin{equation}\label{eq:bdE2}
\| \mathcal{E}_{2} \zeta \| \leq C\|\hat V\|_{1} N^{-\frac{2}{3}} \| \Ncal^{\frac{3}{2}}\zeta \|\;.
\end{equation}
Consider the last term of \eqref{eq:start}. To control the difference $Q_{\textnormal{B}} - Q_{\textnormal{B}}^{\mathcal{R}}$ we use that, as discussed in the proof of  \cite[Lemma 4.1]{BNPSS2},
\begin{equation}
b(k) = b^{\mathcal{R}}(k) + \mathfrak{r}^{\mathcal{R}}(k)\nonumber
\end{equation}
where
\begin{equation}\label{eq:tauU}
b^{\mathcal{R}}(k) = \sum_{\alpha \in \mathcal{I}_{k}} n_{\alpha}(k) c_{\alpha}(k)\;,\qquad \mathfrak{r}^{\mathcal{R}}(k) = \sum_{p \in U} a_{p-k} a_{p}\;,
\end{equation}
with 
\[
 U=(B_\textnormal{F}^c \cap (B_\textnormal{F}+k)) \backslash \bigcup_{\alpha\in \mathcal{I}_k} (B_\alpha \cap (B_\alpha+k))\;.
\]
The contributions from $b^{\mathcal{R}}(k)$ will form $Q_\textnormal{B}^\mathcal{R}$, the contributions from $\mathfrak{r}^{\mathcal{R}}(k)$ are to be estimated as a small error. 
Proceeding similarly to \eqref{eq:bk} and \eqref{eq:bk-bk*}, we find that 
\begin{equation}\label{eq:bRest}
\| b^{\mathcal{R}}(k) \zeta \| +  \| b^{\mathcal{R}}(k)^{*} \zeta \| \leq CN^{\frac{1}{3}} \| (\Ncal+1)^{\frac{1}{2}}\zeta \|\;.
\end{equation}
We turn to the operators $\mathfrak{r}^{\mathcal{R}}(k)$. We rewrite $U = Y \cup (U \setminus Y)$, with $Y$ the set of lattice points such that
\begin{equation}
Y = \big\{p\in B^{\textnormal{c}}_{\textnormal{F}} \cap (B_{\textnormal{F}} + k) \mid e(p) + e(p-k) \leq 4 N^{-\frac{1}{3} - \delta}\big\}\;.
\end{equation}
The condition $e(p) + e(p-k) \leq 4 N^{-\frac{1}{3}-\delta}$, thanks to $\lvert p \rvert = \Ocal(N^{\frac{1}{3}})$ and $\lvert k \rvert = \Ocal(1)$, implies $\hat{p}\cdot\hat{k} \leq C N^{-\delta}$. This is a ribbon of width $N^{-\delta + \frac{1}{3}}$ around the equator of the Fermi sphere, containing at most $\Ocal(N^{-\delta + \frac{1}{3}} \times N^{\frac{1}{3}})$ lattice points. The set $U\setminus Y$ consists of the corridors between the remaining patches. Accordingly, the number of lattice points in $U\setminus Y$ is
\begin{equation}\label{eq:UY}
|U\setminus Y| \leq CN^{\frac{1}{3}} M^{\frac{1}{2}}\;.
\end{equation}
Correspondingly, we split
\begin{equation}\label{eq:rsplit}
\mathfrak{r}^{\mathcal{R}}(k) = \mathfrak{r}_{Y}^{\mathcal{R}}(k) + \mathfrak{r}_{U\setminus Y}^{\mathcal{R}}(k)\;,
\end{equation}
where the two operators are defined as in \cref{eq:tauU}, replacing the index set $U$ with $Y$ or with $U\setminus Y$, respectively. By the Cauchy--Schwarz inequality 
\begin{equation}\label{eq:tbd}
\| \mathfrak{r}_{U\setminus Y}^{\mathcal{R}}(k) \zeta \| \leq \sum_{p\in U \setminus Y} \| a_{p} \xi \| \leq CN^{\frac{1}{6}} M^{\frac{1}{4}} \| \Ncal^{\frac{1}{2}} \zeta \|\;.
\end{equation}
To estimate $\mathfrak{r}_{U\setminus Y}^{\mathcal{R}}(k)^{*}$ we write
\begin{equation}
\| \mathfrak{r}_{U\setminus Y}^{\mathcal{R}}(k)^* \zeta \|^{2} = | \langle \zeta, [ \mathfrak{r}_{U\setminus Y}^{\mathcal{R}}(k), \mathfrak{r}_{U\setminus Y}^{\mathcal{R}}(k)^* ] \zeta \rangle| + \| \mathfrak{r}_{U\setminus Y}^{\mathcal{R}}(k) \zeta \|^{2}\;.
\end{equation}
The second term is estimated as in \cref{eq:tbd}. The first term can be estimated by computing the commutator and recalling the estimate \cref{eq:UY} on the number of lattice points in $U\setminus Y$.
We get
\begin{equation}\label{eq:tbd2}
\| \mathfrak{r}_{U\setminus Y}^{\mathcal{R}}(k)^{*} \zeta \| \leq CN^{\frac{1}{6}} M^{\frac{1}{4}} \| (\Ncal + 1)^{\frac{1}{2}} \zeta \|\;.
\end{equation} 
For $\mathfrak{r}_{Y}^{\mathcal{R}}(k)$, by the definition of $Y$ we get
\[
\| \mathfrak{r}_{Y}^{\mathcal{R}}(k) \zeta \| \leq \sum_{p\in B_\textnormal{F}^c \cap (B_\textnormal{F}+k)} \frac{2 N^{-\frac{1}{6}-\frac{\delta}{2}}}{\sqrt{e(p)+e(p-k)}} \|a_p\zeta\|\;.
\]
In combination with the Cauchy--Schwarz inequality and \cref{eq:HPR20}
 \begin{align}\label{eq:tbd3}
\| \mathfrak{r}_{Y}^{\mathcal{R}}(k) \zeta \| &\leq 2 N^{-\frac{1}{6}-\frac{\delta}{2}} \Bigg( \sum_{p\in B_\textnormal{F}^c \cap (B_\textnormal{F}+k)} \frac{1}{e(p)+e(p-k)}\Bigg)^{1/2}  \Bigg( \sum_{p\in B_\textnormal{F}^c \cap (B_\textnormal{F}+k)} \|a_p\zeta\|^2\Bigg)^{1/2} \nonumber\\
&\le C N^{-\frac{1}{6}-\frac{\delta}{2}} N^{\frac1 2} \|\Ncal^{\frac{1}{2}} \zeta\| = C N^{\frac{1}{3}-\frac{\delta}{2}}  \|\Ncal^{\frac{1}{2}} \zeta\|\;. 
\end{align}
(The bound \cref{eq:HPR20} yields a constant that depends on $k$, but since only $k$ in the compact set $\supp \hat{V}$ is relevant here, we can take the maximum with respect to $k$.)
Concerning $\mathfrak{r}_{Y}^{\mathcal{R}}(k)^{*}$, we use
\[
\| \mathfrak{r}_{Y}^{\mathcal{R}}(k)^{*} \zeta \|^{2} \leq |\langle \zeta, [ \mathfrak{r}_{Y}^{\mathcal{R}}(k), \mathfrak{r}_{Y}^{\mathcal{R}}(k)^{*} ] \zeta \rangle| + \|\mathfrak{r}_{Y}^{\mathcal{R}}(k) \zeta  \|^{2}\;.
\]
The second term is estimated as in \cref{eq:tbd3}. The first term can be bounded by computing the commutator and using the estimate $\lvert Y\rvert \leq C N^{\frac{2}{3} - \delta}$ for the number of lattice points in $Y$. Therefore
\begin{equation}\label{eq:tbd4}
\| \mathfrak{r}_{Y}^{\mathcal{R}}(k)^{*} \zeta \| \leq C N^{\frac{1}{3} - \frac{\delta}{2}}  \| (\Ncal + 1)^{\frac{1}{2}} \zeta \|\;.
\end{equation} 
In summary, from \eqref{eq:tbd}, \eqref{eq:tbd2}, \eqref{eq:tbd3} and \eqref{eq:tbd4} we find that 
\begin{align} \label{eq:Tau-rest}
\| \mathfrak{r}^{\mathcal{R}}(k) \zeta \| +\| \mathfrak{r}^{\mathcal{R}}(k)^* \zeta \| \leq C(N^{\frac{1}{3} - \frac{\delta}{2}} + N^{\frac{1}{6}} M^{\frac{1}{4}}) \| (\Ncal +1)^{\frac{1}{2}} \zeta \| \;.
\end{align}

With this, we are ready to estimate $\| (Q_{\textnormal{B}} - Q^{\mathcal{R}}_{\textnormal{B}}) \zeta\|$. Note that the difference $Q_{\textnormal{B}} - Q^{\mathcal{R}}_{\textnormal{B}}$ is given by a sum of terms containing at least one operator $\mathfrak{r}^{\mathcal{R}}(k)^{\natural}$. Moreover, similar to $b(k)$, both  $b^{\mathcal{R}}(k)$ and $\mathfrak{r}^{\mathcal{R}}(k)$ annihilate two fermions (i.\,e., $b^{\mathcal{R}}(k)\Ncal= (\Ncal+2)b^{\mathcal{R}}(k)$ and $\mathfrak{r}^{\mathcal{R}}(k)\Ncal= (\Ncal+2)\mathfrak{r}^{\mathcal{R}}(k)$).  The bounds \cref{eq:bRest} and \cref{eq:Tau-rest} imply (to simplify the estimate recall that $M^{\frac{1}{4}} \ll N^{\frac{1}{6}}$ and use $(\Ncal+3) \leq C (\Ncal+1)$)
\begin{align} \label{eq:Q-QR}
\| (Q_{\textnormal{B}} - Q^{\mathcal{R}}_{\textnormal{B}}) \zeta \| &\leq \frac{C\norm{\hat{V}}_{1}}{N} N^{\frac{1}{3}} \Big( N^{\frac{1}{3} - \frac{\delta}{2}} +N^{\frac{1}{6}} M^{\frac{1}{4}}  \Big) \| (\Ncal + 1) \zeta \| \nonumber\\
&\leq C \norm{\hat{V}}_{1} \hbar \Big( N^{-\frac{\delta}{2}} + N^{-\frac{1}{6}} M^{\frac 14} \Big) \| (\Ncal + 1)\zeta \|\;. 
\end{align}
Inserting \eqref{eq:bdX}, \eqref{eq:E1}, \eqref{eq:bdE2}  and \eqref{eq:Q-QR} in \eqref{eq:start} we obtain the desired estimate.
\end{proof}

\section{Linearization of Kinetic Term}\label{sec:lin}
Our next task will be to approximate the fermionic kinetic energy $\Hbb_0$ with the bosonized kinetic energy $\mathbb{D}_\B$ defined in \cref{eq:Dbb}. This is the step where the division into patches is needed. 

First, we show that for any state $\zeta\in \mathcal{F}$ with few excitations the norm $\| (\Hbb_0-\mathbb{D}_{\textnormal{B}}) \zeta\|$ is essentially invariant under the Bogoliubov transformation $\zeta \mapsto T\zeta$.  

\begin{lem}[Approximation of kinetic energy, part I]\label{prp:kin2} There exists $C > 0$ such that for all $\zeta\in \mathcal{F}$ we have
\begin{align}\label{eq:HD}
 &\Big| \big\| (\mathbb{H}_{0} - \mathbb{D}_{\textnormal{B}}) T \zeta\big\| - \big\| (\mathbb{H}_{0} - \mathbb{D}_{\textnormal{B}}) \zeta\big\| \Big| \nonumber\\
 & \leq C \hbar \Big( M^{-\frac{1}{2}} \| (\Ncal + 1) \zeta \|+  M N^{-\frac{2}{3} + \delta} \|(\Ncal + 1)^{2}  \zeta \| \Big)\;.  
\end{align}
\end{lem}
\begin{proof} Recall that $T_\lambda = \exp(\lambda B)$, with $B$ the operator defined in \cref{eq:B}. Then we get
\begin{align}\label{eq:comm2}
&\frac{\di}{\di\lambda} \langle T_{\lambda} \zeta, (\mathbb{H}_{0} - \mathbb{D}_{\textnormal{B}})^{2} T_{\lambda}\zeta\rangle  = \langle  T_{\lambda} \zeta, [(\mathbb{H}_{0} - \mathbb{D}_{\textnormal{B}})^{2}, B] T_{\lambda}\zeta\rangle \nonumber\\
& = 2 \Re  \langle  T_{\lambda}\zeta,  (\mathbb{H}_{0} - \mathbb{D}_{\textnormal{B}}) [(\mathbb{H}_{0} - \mathbb{D}_{\textnormal{B}}), B] T_{\lambda}\zeta\rangle\;.  
\end{align}
Here the self--adjointness of $\mathbb{H}_{0} - \mathbb{D}_{\textnormal{B}}$ and anti--self--adjointness of $B$ have been used in the second equality. 
The commutator $[(\mathbb{H}_{0} - \mathbb{D}_{\textnormal{B}}), B]$ can be evaluated with the aid of \cref{lem:boskin}. Using 
\begin{equation}
[ (\mathbb{H}_{0} - \mathbb{D}_{\textnormal{B}}), c^{*}_{\alpha}(k) ] = \hbar\mathfrak{E}^{\textnormal{lin}}_{\alpha}(k)^{*} - \hbar\mathfrak{E}^{\textnormal{B}}_{\alpha}(k)^{*}\;,
\end{equation}
we get
\begin{align*}
&[(\mathbb{H}_{0} - \mathbb{D}_{\textnormal{B}}), B] \\
&= \frac{\hbar}{2} \sum_{k\in \north} \sum_{\alpha, \beta \in \mathcal{I}_{k}} K(k)_{\alpha,\beta} \Big( c^{*}_{\alpha}(k) (\mathfrak{E}^{\textnormal{lin}}_{\beta}(k)^{*} - \mathfrak{E}^{\textnormal{B}}_{\beta}(k)^{*}) +  (\mathfrak{E}^{\textnormal{lin}}_{\alpha}(k)^{*} - \mathfrak{E}^{\textnormal{B}}_{\alpha}(k)^{*}) c^{*}_{\beta}(k)\Big) + \hc\nonumber
\end{align*}
In combination with $|K(k)_{\alpha,\beta}|\le CM^{-1}$ from Lemma \ref{lem:K}, we deduce from \eqref{eq:comm2} that 
\begin{align}\label{eq:comm2a}
&\Big|\frac{\di}{\di\lambda} \Big( \| (\mathbb{H}_{0} - \mathbb{D}_{\textnormal{B}}) T_{\lambda}\zeta\|^2 \Big)  \Big|  \le C\hbar\, M^{-1} \sum_{k\in \north}(  I_1(k)+I_2(k)+I_3(k)+I_4(k)) 
\end{align}
where
\begin{align*} 
I_1(k)& := \sum_{\alpha, \beta\in \mathcal{I}_{k}} \Big| \langle  T_{\lambda}\zeta,  (\mathbb{H}_{0} - \mathbb{D}_{\textnormal{B}})  c^{*}_{\alpha}(k) (\mathfrak{E}^{\textnormal{lin}}_{\beta}(k)^{*} - \mathfrak{E}^{\textnormal{B}}_{\beta}(k)^{*}) T_{\lambda}\zeta\rangle \Big|\;, \\
I_2(k)& := \sum_{\alpha, \beta\in \mathcal{I}_{k}} \Big| \langle  T_{\lambda}\zeta,  (\mathbb{H}_{0} - \mathbb{D}_{\textnormal{B}})  (\mathfrak{E}^{\textnormal{lin}}_{\alpha}(k)^{*} - \mathfrak{E}^{\textnormal{B}}_{\alpha}(k)^{*}) c^{*}_{\beta}(k) T_{\lambda}\zeta\rangle \Big|\;, \\
I_3(k)& := \sum_{\alpha, \beta\in \mathcal{I}_{k}} \Big| \langle  T_{\lambda}\zeta,  (\mathbb{H}_{0} - \mathbb{D}_{\textnormal{B}})  (\mathfrak{E}^{\textnormal{lin}}_{\beta}(k)  - \mathfrak{E}^{\textnormal{B}}_{\beta}(k) ) c_{\alpha}(k)  T_{\lambda}\zeta\rangle \Big|\;, \\
I_4(k)& := \sum_{\alpha, \beta\in \mathcal{I}_{k}} \Big|  \langle  T_{\lambda}\zeta,  (\mathbb{H}_{0} - \mathbb{D}_{\textnormal{B}}) c_{\beta}(k)  (\mathfrak{E}^{\textnormal{lin}}_{\alpha}(k) - \mathfrak{E}^{\textnormal{B}}_{\alpha}(k))  T_{\lambda}\zeta\rangle \Big| \;.
\end{align*}
We estimate the right--hand side of \eqref{eq:comm2a}.
Since the operators $\mathfrak{E}^{\textnormal{lin}}_{\alpha}(k)$, $\mathfrak{E}^{\textnormal{B}}_{\alpha}(k)$, $c_{\alpha}(k)$ all annihilate two fermions, we get
\begin{align*}  
I_1(k)&= \sum_{\alpha, \beta\in \mathcal{I}_{k}} \Big| \langle  \mathfrak{E}^{\textnormal{lin}}_{\beta}(k) (\Ncal+3)^{-1/2} c_{\alpha}(k) (\Ncal+1)^{-1/2} (\mathbb{H}_{0} - \mathbb{D}_{\textnormal{B}})  T_{\lambda}\zeta,   (\Ncal+5) T_{\lambda}\zeta\rangle \\
&\qquad\qquad\quad - \langle  \mathfrak{E}^{\textnormal{B}}_{\beta}(k) (\Ncal+3)^{-3/2} c_{\alpha}(k) (\Ncal+1)^{-1/2} (\mathbb{H}_{0} - \mathbb{D}_{\textnormal{B}})  T_{\lambda}\zeta,   (\Ncal+5)^2 T_{\lambda}\zeta\rangle \Big|\;. 
\end{align*}
Then by the Cauchy--Schwarz inequality and the bounds \cref{eq:sumc,eq:sumfrakE} we get
\begin{align}   \label{eq:I1k-final}
I_1(k) &\le \sum_{\alpha, \beta\in \mathcal{I}_{k}} \|  \mathfrak{E}^{\textnormal{lin}}_{\beta}(k) (\Ncal+3)^{-1/2} c_{\alpha}(k) (\Ncal+1)^{-1/2} (\mathbb{H}_{0} - \mathbb{D}_{\textnormal{B}})  T_{\lambda}\zeta \| \, \| (\Ncal+5) T_{\lambda}\zeta \| \nonumber\\
&\qquad + \sum_{\alpha, \beta\in \mathcal{I}_{k}} \| \mathfrak{E}^{\textnormal{B}}_{\beta}(k) (\Ncal+3)^{-3/2} c_{\alpha}(k) (\Ncal+1)^{-1/2} (\mathbb{H}_{0} - \mathbb{D}_{\textnormal{B}})  T_{\lambda}\zeta \| \, \|  (\Ncal+5)^2 T_{\lambda}\zeta  \| \nonumber \\
&\le  C\sum_{\alpha\in \mathcal{I}_{k}} \|  c_{\alpha}(k) (\Ncal+1)^{-1/2} (\mathbb{H}_{0} - \mathbb{D}_{\textnormal{B}})  T_{\lambda}\zeta \| \, \| (\Ncal+5) T_{\lambda}\zeta \| \nonumber\\
&\qquad + CM^{\frac{3}{2}} N^{-\frac{2}{3} +\delta}\sum_{\alpha\in \mathcal{I}_{k}} \| c_{\alpha}(k) (\Ncal+1)^{-1/2} (\mathbb{H}_{0} - \mathbb{D}_{\textnormal{B}})  T_{\lambda}\zeta \| \, \|  (\Ncal+5)^2 T_{\lambda}\zeta  \| \nonumber\\
&\le C  \| (\mathbb{H}_{0} - \mathbb{D}_{\textnormal{B}})  T_{\lambda}\zeta \| \Big( M^{\frac{1}{2}} \| (\Ncal+1) T_{\lambda}\zeta \| + M^{2} N^{-\frac{2}{3} +\delta}\|  (\Ncal+1)^2 T_{\lambda}\zeta  \| \Big)\;. 
\end{align}
The terms $I_2(k)$ to $I_4(k)$ can be estimated in a similar way.

Employing \eqref{eq:I1k-final} in \eqref{eq:comm2a} and using \cref{lem:stabN} we get
\begin{align*} 
&\Big|\frac{\di}{\di\lambda} \Big( \| (\mathbb{H}_{0} - \mathbb{D}_{\textnormal{B}}) T_{\lambda}\zeta\|^2 \Big)  \Big| \\
&\le \| (\mathbb{H}_{0} - \mathbb{D}_{\textnormal{B}})  T_{\lambda}\zeta \| C\hbar  \Big( M^{-\frac{1}{2}} \| (\Ncal+1) T_{\lambda}\zeta \| + M N^{-\frac{2}{3} +\delta}\|  (\Ncal+1)^2 T_{\lambda}\zeta  \| \Big)\\
&\le  \| (\mathbb{H}_{0} - \mathbb{D}_{\textnormal{B}})  T_{\lambda}\zeta \|   C\hbar \Big( M^{-\frac{1}{2}} \| (\Ncal+1) \zeta \| + M N^{-\frac{2}{3} +\delta}\|  (\Ncal+1)^2 \zeta  \| \Big)
\end{align*}
for all $\lambda\in [0,1]$. This implies
 \begin{align*} 
\Big|\frac{\di}{\di\lambda}   \| (\mathbb{H}_{0} - \mathbb{D}_{\textnormal{B}}) T_{\lambda}\zeta\|  \Big| \le C\hbar \Big( M^{-\frac{1}{2}} \| (\Ncal+1) \zeta \| + M N^{-\frac{2}{3} +\delta}\|  (\Ncal+1)^2 \zeta  \| \Big)\;.
\end{align*}
Integrating over $\lambda\in [0,1]$ leads to the desired inequality \cref{eq:HD}.
\end{proof}
Next, we show that $\| (\mathbb{H}_{0} - \mathbb{D}_{\textnormal{B}}) \zeta \|$ is small for our choice of $\zeta$. 

\begin{lem}[Approximation of kinetic energy, part II]\label{prp:almbos} For $i=1,\ldots, m$, let $\varphi_i \in \bigoplus_{k\in \north} \ell^2(\mathcal{I}_k)$ be normalized as in \cref{eq:139}. Then  
\begin{equation}
\| (\mathbb{H}_{0} - \mathbb{D}_{\textnormal{B}}) c^{*}(\varphi_{1}) \cdots c^{*}(\varphi_{m}) \Omega \| \leq C \hbar m^2\sqrt{(2m-1)!!}  \big(M^{-\frac{1}{2}} + M^{\frac{3}{2}} N^{-\frac{2}{3} + \delta}\big)\;. \label{eq:57}
\end{equation}
\end{lem}
\begin{proof} Since $(\mathbb{H}_{0} - \mathbb{D}_{\textnormal{B}}) \Omega = 0$ we can write
\begin{align}
&(\mathbb{H}_{0} - \mathbb{D}_{\textnormal{B}}) c^{*}(\varphi_{1}) \cdots c^{*}(\varphi_{m}) \Omega=  [(\mathbb{H}_{0} - \mathbb{D}_{\textnormal{B}}), c^{*}(\varphi_{1}) \cdots c^{*}(\varphi_{m})] \Omega \nonumber \\
&=\sum_{i=1}^{m} c^{*}(\varphi_{1}) \cdots c^{*}(\varphi_{i-1}) [ (\mathbb{H}_{0} - \mathbb{D}_{\textnormal{B}}), c^{*}(\varphi_{i}) ] c^{*}(\varphi_{i+1})\cdots c^{*}(\varphi_{m}) \Omega\;.  \label{eq:bigcommHD}
\end{align}
Recalling \cref{eq:linearization}, the commutator on the right--hand side of \eqref{eq:bigcommHD} is
\begin{align}\label{eq:comm4}
[ (\mathbb{H}_{0} - \mathbb{D}_{\textnormal{B}}), c^{*}(\varphi_{i}) ] &= \sum_{k\in \north} \sum_{\alpha \in \mathcal{I}_{k}} (\varphi_{i}(k))_\alpha [ (\mathbb{H}_{0} - \mathbb{D}_{\textnormal{B}}), c^{*}_{\alpha}(k) ] \nonumber\\
&= \hbar\sum_{k\in \north} \sum_{\alpha \in \mathcal{I}_{k}} (\varphi_{i}(k))_\alpha (\mathfrak{E}^{\textnormal{lin}}_{\alpha}(k)^{*} - \mathfrak{E}^{\textnormal{B}}_{\alpha}(k)^{*})\;.
\end{align}
Recall the bounds \eqref{eq:sumEE}; thus for all $\xi \in \fock$ we have
\[
\Big\| [ (\mathbb{H}_{0} - \mathbb{D}_{\textnormal{B}}), c^{*}(\varphi_{i}) ] \xi\Big\| \le C\hbar \left( M^{-\frac{1}{2}} \|(\Ncal+1)^{\frac{1}{2}}\xi\| +  M^{\frac{3}{2}} N^{-\frac{2}{3} + \delta}  \| (\Ncal + 1)^{\frac{3}{2}} \xi\| \right)\;.
\]
Recall \cref{eq:Nccomm}; just as $c^{*}(\varphi_j)$ creates two fermions, so does the operator $[ (\mathbb{H}_{0} - \mathbb{D}_{\textnormal{B}}), c^{*}(\varphi_{i}) ]$ which therefore also has the commutator $[\Ncal,[ (\mathbb{H}_{0} - \mathbb{D}_{\textnormal{B}}), c^{*}(\varphi_{i}) ] ] = 2 [ (\mathbb{H}_{0} - \mathbb{D}_{\textnormal{B}}), c^{*}(\varphi_{i}) ]$. Recalling also the simple bound $\|c^*(\varphi_j)\xi\| \le \| (\Ncal +1)^{\frac{1}{2}} \xi\|$ we obtain  
\begin{align}
& \| c^{*}(\varphi_{1}) \cdots c^{*}(\varphi_{i-1}) [ (\mathbb{H}_{0} - \mathbb{D}_{\textnormal{B}}), c^{*}(\varphi_{i}) ] c^{*}(\varphi_{i+1})\cdots c^{*}(\varphi_{m}) \Omega \| \nonumber\\
&\leq \| \prod_{j=1}^{i-1} (\Ncal + 1 + 2(i - 1 -j))^{1/2} [ (\mathbb{H}_{0} - \mathbb{D}_{\textnormal{B}}), c^{*}(\varphi_{i}) ] c^{*}(\varphi_{i+1})\cdots c^{*}(\varphi_{m}) \Omega\| \nonumber\\
& \leq C \Big(\prod_{j=1}^{i-1} (2(m - j) + 1)^{\frac{1}{2}}\Big) \| [ (\mathbb{H}_{0} - \mathbb{D}_{\textnormal{B}}), c^{*}(\varphi_{i}) ] c^{*}(\varphi_{i+1})\cdots c^{*}(\varphi_{m}) \Omega\| \nonumber \\
&\le C \hbar \Big(\prod_{j=1}^{i-1} (2(m - j) + 1)^{\frac{1}{2}}\Big) \Big( M^{-\frac{1}{2}} \| (\Ncal+1)^{\frac{1}{2}} c^{*}(\varphi_{i+1})\cdots c^{*}(\varphi_{m}) \Omega\| \label{eq:thirdlast}\\
&\hspace{13em} +M^{\frac{3}{2}} N^{-\frac{2}{3} + \delta }\| (\Ncal+1)^{\frac{3}{2}} c^{*}(\varphi_{i+1})\cdots c^{*}(\varphi_{m}) \Omega\| \Big) \label{eq:seclast}\\
&\le C \hbar m\sqrt{(2m-1)!!} (M^{-\frac{1}{2}} + M^{\frac{3}{2}} N^{-\frac{2}{3} + \delta })\;. \nonumber
\end{align}
Note that in \cref{eq:seclast} there is an additional factor $(\Ncal+1)$ compared to \cref{eq:thirdlast}, which was bounded by an additional $m$ on the last line.
A further factor of $m$ in \cref{eq:57} is due to the sum over $i$ in \cref{eq:bigcommHD}.
\end{proof}

\section{Diagonalization of Effective Hamiltonian} \label{sec:diag}

Recall the definition of $h_\textnormal{eff}(k)$ in \cref{eq:heff}, and consider the effective Hamiltonian
\begin{equation}
\mathbb{D}_{\textnormal{B}} + Q^{\mathcal{R}}_{\textnormal{B}} = \sum_{k\in \north} 2\hbar \kappaf \lvert k\rvert h_{\textnormal{eff}}(k) \;.
\end{equation}
\begin{lem}[Bogoliubov diagonalization] \label{lem:diag}
Let $\widetilde E_N^\textnormal{RPA}$ be given in \eqref{eq:GSE-RPA}, and $T$ the unitary transformation defined in \cref{eq:B}. For all $s\ge 0$ we have
\begin{align}
&\| \Big(T^*(\mathbb{D}_{\textnormal{B}} + Q^{\mathcal{R}}_{\textnormal{B}}  )T - \mathfrak{H}_{\textnormal{exc}} - \widetilde E_N^\textnormal{RPA} \Big)\xi_{s} \| \le C \left( M^{\frac{3}{2}}N^{-\frac{2}{3} + \delta} (m+1)^2 + \hbar N^{-\frac{2}{3} + \delta} m \right)\|\xi_s\|\;.   \nonumber
\end{align}
\end{lem}

\begin{proof}
As already shown in \cite[Eq.~(10.5)]{BNPSS2}, by Lemma \ref{lem:Bog} we get
\begin{align}\label{eq:ThT}
T^{*} h_{\textnormal{eff}}(k) T &= \sum_{\alpha, \beta \in \mathcal{I}_{k}} \big( D(k) + W(k) \big)_{\alpha, \beta} (\tilde c^{*}_{\alpha}(k) + \mathfrak{E}^{*}_{\alpha}(1,k)) (\tilde c_{\beta}(k) + \mathfrak{E}_{\beta}(1,k))\\
&\quad  + \frac{1}{2} \sum_{\alpha, \beta\in \mathcal{I}_{k}} \Big[ \widetilde{W}(k)_{\alpha, \beta} ( \tilde c^{*}_{\alpha}(k) + \mathfrak{E}^{*}_{\alpha}(1,k) ) (\tilde c^{*}_{\beta}(k) + \mathfrak{E}^{*}_{\beta}(1,k) ) + \hc \Big]\;,\nonumber
\end{align}
where
\begin{equation}
\tilde c_{\alpha}(k) = \sum_{\beta \in \mathcal{I}_{k}} \cosh( K(k) )_{\alpha, \beta} c_{\beta}(k) + \sum_{\beta\in \mathcal{I}_{k}} \sinh ( K(k) )_{\alpha, \beta} c^{*}_{\beta}(k)\;.
\end{equation}
The  transformed operators satisfy, for $f \in \ell^2(\Ical_k)$,
\begin{equation}
\| \tilde c(f) \zeta\| \leq C\|f\|_{\ell^{2}} \| (\Ncal + 1)^{\frac{1}{2}} \zeta \|\;,\qquad \| \tilde c(f)^{*} \zeta\| \leq C\|f\|_{\ell^{2}} \| (\Ncal + 1)^{\frac{1}{2}} \zeta \|\;.
\end{equation}
This is easily seen using the formula
\begin{align}
\tilde c(f) = \sum_{\alpha\in \mathcal{I}_{k}} \overline{f_{\alpha}} \tilde c_{\alpha}(k) = c\big(\cosh(K(k)) f\big) + c^{*}\big(\sinh(K(k))\overline{f}\big)
\end{align}
so that
\begin{align}\label{eq:tildecbd}
\| \tilde c(f) \psi \| &\leq \|\cosh(K(k)) f\|_{2} \| \Ncal^{\frac{1}{2}}\psi \| + \|\sinh(K(k)) f \|_{2} \| (\Ncal + 1)^{\frac{1}{2}}\psi \|\nonumber\\
&\leq C\|f\|_{2} \| (\Ncal + 1)^{\frac{1}{2}}\psi \|\;.
\end{align}
 Here we used that according to \cref{lem:K} we have 
\begin{equation} 
 \begin{split}
 \| \cosh(K(k)) \|_{\textnormal{op}} & \leq \norm{\id}_\textnormal{op} + \norm{\cosh(K(k)) - \id}_\textnormal{HS}  \leq C\;, \\
\|\sinh(K(k))\|_{\textnormal{op}} & \leq \|\sinh(K(k))\|_{\textnormal{HS}} \leq C\;. 
 \end{split}
\end{equation}
The same bound as \cref{eq:tildecbd} holds for $\tilde c^{*}(f)$.

\paragraph{Controlling the error terms of the diagonalization.} We proceed to estimate 
\[
\| (T^{*} h_{\textnormal{eff}}(k) T - h^{\textnormal{diag}}_{\textnormal{eff}}(k)) \xi_{s} \| 
\]
where 
\begin{equation}
h^{\textnormal{diag}}_{\textnormal{eff}}(k) = \sum_{\alpha, \beta \in \mathcal{I}_{k}} \big( D(k) + W(k) \big)_{\alpha, \beta} \tilde c^{*}_{\alpha}(k) \tilde c_{\beta}(k)  + \frac{1}{2} \widetilde{W}(k)_{\alpha, \beta} \big( \tilde c^{*}_{\alpha}(k)\tilde c^{*}_{\beta}(k) + \hc \big) \Big]\;. 
\end{equation}
Every summand of the difference $T^{*} h_{\textnormal{eff}}(k) T - h^{\textnormal{diag}}_{\textnormal{eff}}(k)$ contains at least one factor of the error term $\mathfrak{E}^{\natural}(1,k)$. To estimate it, recall the bounds \cref{eq:DWWtbound}. Then, using \eqref{eq:tildecbd} and \cref{prp:NkE}, we get, for instance,
\begin{align}\label{eq:NE}
\Big\| \sum_{\alpha, \beta \in \mathcal{I}_{k}} \widetilde{W}(k)_{\alpha, \beta} \tilde c^{*}_{\alpha}(k) \mathfrak{E}^{*}_{\beta}(1,k) \xi_{s} \Big\| & \leq CM^{-\frac{1}{2}} \sum_{\beta \in \Ik} \| (\Ncal + 1)^{1/2} \mathfrak{E}^{*}_{\beta}(1,k) \xi_{s} \| \nonumber \\
&\leq CM^{\frac{1}{2}}N^{-\frac{2}{3} + \delta} \| (\Ncal + M)^{\frac{1}{2}}(\Ncal + 1)^{\frac{3}{2}} \xi_{s} \| \nonumber\\
& \leq CMN^{-\frac{2}{3} + \delta} (2m+1)^{2} \norm{\xi_s}\;.
\end{align}
The other terms contributing to $T^{*} h_{\textnormal{eff}}(k) T - h^{\textnormal{diag}}_{\textnormal{eff}}(k)$ involving only $\widetilde{W}(k)$ are bounded in the same way. Next, consider the terms involving $\widetilde{D}(k) := D(k) + W(k)$. For instance
\begin{align}
\Big\| \sum_{\alpha, \beta \in \mathcal{I}_{k}} \widetilde{D}(k) _{\alpha,\beta}\tilde c^{*}_{\alpha}(k) \mathfrak{E}_{\beta}(1,k) \xi_{s}\Big\|
&\leq \sum_{\beta \in \mathcal{I}_{k}} C\| (\Ncal + 1)^{\frac{1}{2}} \mathfrak{E}_{\beta}(1,k) \xi_{s} \|\nonumber\\
&\leq CM N^{-\frac{2}{3} + \delta} \| (\Ncal + M)^{\frac{1}{2}} (\Ncal + 1)^{\frac{3}{2}} \xi_{s} \| \nonumber\\
&\leq CM^{\frac{3}{2}}N^{-\frac{2}{3} + \delta} (m+1)^{2} \norm{\xi_s}\;.
\end{align}
In the first step we used that $\sum_{\alpha \in \Ik} \lvert \widetilde{D}(k)_{\alpha,\beta} \rvert^2 \leq C$ uniform in $\beta$.  In the second step we used \cref{prp:NkE}. All the other terms contributing to $T^{*} h_{\textnormal{eff}}(k) T - h^{\textnormal{diag}}_{\textnormal{eff}}(k)$ involving $\widetilde{D}(k)$ are estimated in the same way. In conclusion
\begin{equation} \label{eq:diag-1}
\| \Big(T^{*} h_{\textnormal{eff}}(k) T - h^{\textnormal{diag}}_{\textnormal{eff}}(k)\Big) \xi_{s} \| \leq CM^{\frac{3}{2}}N^{-\frac{2}{3} + \delta} (m+1)^{2} \|\xi_s\|\;.
\end{equation}
\paragraph{Error term due to normal-ordering.} As computed in  \cite[Section 10]{BNPSS2} 
\begin{equation}
\sum_{k\in \north} 2\hbar \kappaf \lvert k\rvert h^{\textnormal{diag}}_{\textnormal{eff}}(k) = \widetilde E_N^\textnormal{RPA}+ \mathfrak{H}_{\textnormal{exc}} + \mathfrak{E}_{\textnormal{exc}}\;,
\end{equation}
where $\widetilde E_N^\textnormal{RPA}$ is the bosonic ground state energy given in \eqref{eq:GSE-RPA} and 
\begin{align*}
\mathfrak{H}_{\textnormal{exc}} &=  \sum_{k\in \north} 2\hbar \kappaf \lvert k\rvert \sum_{\alpha, \beta \in \mathcal{I}_{k}}  \mathfrak{K}(k)_{\alpha,\beta} c^{*}_{\alpha}(k) c_{\beta}(k)\nonumber\\
\mathfrak{E}_{\textnormal{exc}} &= \sum_{k\in \north} \hbar \kappaf \lvert k\rvert \sum_{\alpha \in \mathcal{I}_{k}} \Big[ 2\sinh (K) (D + W) \sinh (K) \nonumber\\
& \hspace{8.5em} + \cosh (K) \widetilde{W} \sinh (K) + \sinh (K) \widetilde{W} \cosh (K)\Big]_{\alpha, \alpha} \mathcal{E}_{\alpha}(k,k)\;.
\end{align*}
Here the excitation matrix $\Kfrak(k)$ is recovered in the form \cref{eq:curlyKexplicit}.

The error operator $\Ecal_\alpha(k,k)$ is given in \cite[Eq. (5.5)]{BNPSS2}; it commutes with $\Ncal$. From \cite[Eq.~(10.10)]{BNPSS2} we get
$
\pm \mathfrak{E}_{\textnormal{exc}} \le C \hbar N^{-\frac{2}{3} + \delta}  \Ncal$. Since $ \mathfrak{E}_{\textnormal{exc}}$ commutes with $\Ncal$, we get
\begin{align} \label{eq:diag-3}
\| \mathfrak{E}_{\textnormal{exc}} \xi_{s} \| \leq  C \hbar  N^{-\frac{2}{3} + \delta} \| \Ncal \xi_{s} \| \leq C \hbar N^{-\frac{2}{3} + \delta} m \|\xi_s\|\;.
\end{align}
In summary, from \eqref{eq:diag-1} and \eqref{eq:diag-3} we obtain the desired estimate. 
\end{proof}

\section{Proof of Theorem \ref{thm:1}}\label{sec:final}
\begin{proof}[Proof of Theorem \ref{thm:1}.]
It suffices to consider the case $t\ge 0$. Under the assumptions $M\ll N^{\frac{2}{3} -2\delta}$ and $m^3(2m-1)!! \ll N^\delta$ (see Remark (ii) after \cref{thm:1}) we find $M N^{-\frac{2}{3} + \delta} m^3(2m-1)!! \ll 1$.
Therefore, by Lemma \ref{prp:norm}
\[
Z_m \ge \frac{1}{2} \quad \textnormal{and}\quad \lvert \norm{ \xi_s} - 1 \rvert \ll 1\;. 
\]
Moreover, by Lemma \ref{lem:stabN} we have 
\begin{equation} \label{eq:N-xi-s}
\|(\Ncal+1)^r T \xi_s\| \le C_r \|(\Ncal+1)^r  \xi_s\|  \le C_r (2m+1)^r, \quad \forall r\in\Nbb\;.
\end{equation}
We proceed to collect all error estimates. By Lemma \ref{eq:norm-red-0},  for all $t\ge 0$ 
\begin{align} \label{eq:norm-red-0-app}
&\norm{ e^{-i\mathcal{H}_{N} t/\hbar} R T \xi - e^{-i (E^{\textnormal{pw}}_{N} + \widetilde E_N^\textnormal{RPA})t/\hbar} R T \xi_{t} }  \nonumber \\
&\le  \frac{1}{\hbar} \int_{0}^{t}\di s\, \big\| (T^*\mathcal{H}_\textnormal{corr} T - \widetilde E_N^\textnormal{RPA} - \mathfrak{H}_{\textnormal{exc}})  \xi_{s}  \big\| + C m^2\sqrt{(2m-1)!!} M^{\frac{3}{2}} N^{-\frac{2}{3} + \delta} t\;. 
\end{align}
To estimate the right--hand side of \eqref{eq:norm-red-0-app}, we use the triangle inequality 
\begin{align} \label{eq:norm-last-app}
\big\| (T^*\mathcal{H}_\textnormal{corr} T - \widetilde E_N^\textnormal{RPA} - \mathfrak{H}_{\textnormal{exc}})  \xi_{s}  \big\| & \le \big\| T^*(\mathcal{H}_\textnormal{corr} - \Hbb_0 - Q_\textnormal{B}^{\Rcal} ) T \xi_s  \big\| + \big\| T^*(\Hbb_0 - \mathbb{D}_\textnormal{B} ) T \xi_s  \big\| \nonumber\\
& \qquad+ \big\| (T^*(\mathbb{D}_\textnormal{B} +Q_\textnormal{B}^{\Rcal})T - \widetilde E_N^\textnormal{RPA} - \mathfrak{H}_{\textnormal{exc}}  ) \xi_s  \big\|\;.
\end{align}
By Lemma \ref{prp:nonbos} and \eqref{eq:N-xi-s},
\begin{align}\label{eq:nonbos-app}
& \big\| T^* (\mathcal{H}_{\textnormal{corr}} - \Hbb_0 - Q_{\textnormal{B}}^{\mathcal{R}}) T \xi_s \big\| =  \big\| (\mathcal{H}_{\textnormal{corr}} - \Hbb_0 - Q_{\textnormal{B}}^{\mathcal{R}}) T \xi_s \big\| \nonumber\\
 &\leq C \|\hat V\|_{\ell^1} \hbar \Big(N^{-\frac{2}{3}} \|\Ncal^{2} T\xi_s \| + N^{-\frac{1}{3}} \| \Ncal^{\frac{3}{2}} T\xi_s \| + (N^{-\frac{\delta}{2}} +N^{-\frac{1}{6}} M^{\frac{1}{4}} )   \| (\Ncal + 1) T\xi_s \| \Big) \nonumber\\
 &\le  C\hbar ( N^{-\frac{\delta}{2}} + N^{-\frac{1}{6}} M^{\frac{1}{4}}) (m+1)^2\;. 
\end{align} 
By Lemma \ref{prp:kin2} we have
\begin{align*}
& \Big| \big\| T^* (\mathbb{H}_{0} - \mathbb{D}_{\textnormal{B}}) T \xi_s\big\| - \big\| (\mathbb{H}_{0} - \mathbb{D}_{\textnormal{B}}) \xi_s \big\| \Big| \\
&\le C \hbar \Big( M^{-\frac{1}{2}} \| (\Ncal + 1) \xi_s \|+  M N^{-\frac{2}{3} + \delta} \|(\Ncal + 1)^{2} \xi_s \| \Big) \\
&\le C \hbar (m+1)^2 \Big( M^{-\frac{1}{2}}  +  M N^{-\frac{2}{3} + \delta}  \Big)
\end{align*}
and by Lemma \ref{prp:almbos}
\begin{align*}
 \big\| (\Hbb_0 - \mathbb{D}_\textnormal{B} ) \xi_s  \big\| \le C \hbar m^2\sqrt{(2m-1)!!}  (M^{-\frac{1}{2}} + M^{\frac{3}{2}} N^{-\frac{2}{3} + \delta})\;. 
\end{align*}
Thus 
\begin{align}\label{eq:HD-app}
\big\| T^*(\Hbb_0 - \mathbb{D}_\textnormal{B} ) T \xi_s  \big\| \le C \hbar (m+1)^2\sqrt{(2m-1)!!}  (M^{-\frac{1}{2}} + M^{\frac{3}{2}} N^{-\frac{2}{3} + \delta})\;.
\end{align}
Finally, by Lemma \ref{lem:diag}
\begin{align} \label{eq:diag-app}
&\| (T^*(\mathbb{D}_{\textnormal{B}} + Q^{\mathcal{R}}_{\textnormal{B}}  )T - \mathfrak{H}_{\textnormal{exc}} - \widetilde E^\textnormal{RPA}_N )\xi_{s} \| \le CM^{\frac{3}{2}}N^{-\frac{2}{3} + \delta} (m+1)^2  + C\hbar N^{ -\frac{2}{3} + \delta} m\;. 
\end{align}
Inserting \eqref{eq:nonbos-app}, \eqref{eq:HD-app} and \eqref{eq:diag-app} in \eqref{eq:norm-last-app} we conclude that 
\begin{align} \label{eq:norm-last-app-app}
&\big\| (T^*\mathcal{H}_\textnormal{corr} T - \widetilde E_N^\textnormal{RPA} - \mathfrak{H}_{\textnormal{exc}})  \xi_{s}  \big\| \nonumber \\
&\le C \hbar (m+1)^2\sqrt{(2m-1)!!}  (N^{-\frac{\delta}{2}}  + M^{-\frac{1}{2}} + M^{\frac{3}{2}}N^{-\frac{1}{3} + \delta} + N^{-\frac{1}{6}} M^{\frac{1}{4}})\;.
\end{align} 
From \eqref{eq:norm-last-app-app} and \eqref{eq:norm-red-0-app} we obtain
\begin{align} 
&\norm{ e^{-i\mathcal{H}_{N} t/\hbar} R T \xi - e^{-i (E^{\textnormal{pw}}_{N} + \widetilde E_N^\textnormal{RPA})t/\hbar} R T \xi_{t} }  \nonumber \\
&\le  C  (m+1)^2\sqrt{(2m-1)!!}   \Big( N^{-\frac{\delta}{2}}  + M^{-\frac{1}{2}} + M^{\frac{3}{2}}N^{-\frac{1}{3} + \delta} + M^{\frac{1}{4}} N^{-\frac{1}{6}}  \Big)\, \lvert t\rvert\;. \label{eq:endofproof}
\end{align}
To replace $\widetilde{E}_{N}^\textnormal{RPA}$ by ${E}_{N}^\textnormal{RPA}$, we have to take into account the additional error term \cref{eq:RPAdiffs}. Consequently, in addition to \cref{eq:endofproof} there will also be an error term of order
\begin{equation}
\label{eq:lasterr}
(m+1)^2\sqrt{(2m-1)!!} \left(M^{-\frac{1}{4}} N^{\frac{\delta}{2}} + M^{\frac{1}{4}}N^{-\frac{1}{6} + \frac{\delta}{2}} \right)\;.
\end{equation}
Choosing $M=N^{4\delta}$ and $\delta=2/45$ completes the proof of Theorem \ref{thm:1}.
\end{proof}

\section*{Acknowledgements}
  NB was supported by Gruppo Nazionale per la Fisica Matematica (GNFM). RS was supported by the European Research Council (ERC) under the European Union’s Horizon 2020 research and innovation programme (grant agreement No.~694227). PTN was supported by the Deutsche Forschungsgemeinschaft (DFG, German Research Foundation) under Germany's Excellence Strategy (EXC-2111-390814868). MP was supported by the European Research Council (ERC) under the European Union’s Horizon 2020 research and innovation programme (ERC StG MaMBoQ, grant agreement No.~802901). BS was supported by the NCCR SwissMAP, the Swiss National Science Foundation through the Grant ``Dynamical and energetic properties of Bose-Einstein condensates'', and the European Research Council (ERC) under the European Union’s Horizon 2020 research and innovation programme through the ERC-AdG CLaQS (grant agreement No.~834782).


\begin{thebibliography}{BNPSS20a} 
 
\bibitem[Bac92]{Bac92}
Volker Bach.
\newblock Error bound for the {{Hartree}}-{{Fock}} energy of atoms and
  molecules.
\newblock \emph{Communications in Mathematical Physics}, 147(3):527--548, 1992.

\bibitem[BBPPT16]{BBPPT16}
Volker Bach, Sébastien Breteaux, Sören Petrat, Peter Pickl, and Tim
Tzaneteas.
\newblock Kinetic energy estimates for the accuracy of the time-
dependent Hartree–Fock approximation with Coulomb interaction.
\newblock \emph{Journal de Mathématiques Pures et Appliquées}, 105(1):1--30, 2016.

\bibitem[BGGM03]{BGGM03} Claude Bardos, François Golse, Alex D.\ Gottlieb, and Norbert J.\ Mauser.
\newblock Mean Field Dynamics of Fermions and the Time-Dependent Hartree–Fock Equation.
\newblock \emph{Journal de Mathématiques Pures et Appliquées}, 82(6):665--83, 2003.

\bibitem[BGGM04]{BGGM04} Claude Bardos, François Golse, Alex D.\ Gottlieb, and Norbert J.\ Mauser.
\newblock Accuracy of the Time-Dependent Hartree–Fock Approximation for Uncorrelated Initial States.
\newblock \emph{Journal of Statistical Physics}, 115(3/4):1037--55, 2004.


\bibitem[BBCS20]{BBCS20}
Chiara Boccato, Christian Brennecke, Serena Cenatiempo, and Benjamin Schlein.
\newblock The excitation spectrum of {{Bose}} gases interacting through
  singular potentials.
\newblock \emph{Journal of the European Mathematical Society}, 22:2331–2403, July 2020.

\bibitem[BBCS18]{BBCS18}
Chiara Boccato, Christian Brennecke, Serena Cenatiempo, and Benjamin Schlein.
\newblock Complete {{Bose}}\textendash{{Einstein Condensation}} in the
  {{Gross}}\textendash{{Pitaevskii Regime}}.
\newblock \emph{Communications in Mathematical Physics}, 359(3):975--1026, May
  2018.

\bibitem[BBCS19a]{BBCS19a}
Chiara Boccato, Christian Brennecke, Serena Cenatiempo, and Benjamin Schlein.
\newblock Bogoliubov theory in the {{Gross}}-{{Pitaevskii}} limit.
\newblock \emph{Acta Mathematica}, 222(2):219--335, 2019.

\bibitem[BBCS19b]{BBCS19b}
Chiara Boccato, Christian Brennecke, Serena Cenatiempo, and Benjamin Schlein.
\newblock Optimal {{Rate}} for {{Bose}}\textendash{{Einstein Condensation}} in
  the {{Gross}}\textendash{{Pitaevskii Regime}}.
\newblock \emph{Communications in Mathematical Physics}, September 2019.

\bibitem[BCS17]{BCS17}
Chiara Boccato, Serena Cenatiempo, and Benjamin Schlein.
\newblock Quantum Many-Body Fluctuations Around Nonlinear Schrödinger Dynamics.
\newblock \emph{Annales Henri Poincaré} 18:113, 2017.

\bibitem[Ben17]{Ben17}
Niels Benedikter.
\newblock Interaction {{Corrections}} to {{Spin}}-{{Wave Theory}} in the
  {{Large}}-{{S Limit}} of the {{Quantum Heisenberg Ferromagnet}}.
\newblock \emph{Mathematical Physics, Analysis and Geometry}, 20(2):5, June
  2017.

\bibitem[Ben19]{Ben19} Niels Benedikter. Bosonic Collective Excitations in Fermi Gases. \emph{Reviews in Mathematical Physics}, 32:2060009 (11 pages), 2020.

\bibitem[BNNS19]{BNNS19}
Christian Brennecke, Phan Thành Nam, Marcin Napiórkowski, and Benjamin Schlein.
Fluctuations of N-particle quantum dynamics around the nonlinear Schrödinger equation. \emph{Annales de l'Institut Henri Poincaré C, Analyse non linéaire}, 
36(5), 1201--1235, 2019.

\bibitem[BNPSS20]{BNPSS} Niels Benedikter, Phan Thành Nam, Marcello Porta, Benjamin Schlein, and Robert Seiringer. Optimal Upper Bound for the Correlation Energy of a Fermi gas in the Mean--Field Regime. \emph{Communications in Mathematical Physics} {374}, 2097--2150 (2020).

\bibitem[BNPSS21]{BNPSS2} Niels Benedikter, Phan Thành Nam, Marcello Porta, Benjamin Schlein, and Robert Seiringer.
\newblock Correlation Energy of a Weakly Interacting Fermi Gas.
\newblock \emph{Inventiones Mathematicae} 225(3):885--979, May 2021. 

  \bibitem[BJPSS16]{BJPSS16}
Niels Benedikter, Vojkan Jak{\v s}i{\'c}, Marcello Porta, Chiara Saffirio, and
  Benjamin Schlein.
\newblock Mean-{{Field Evolution}} of {{Fermionic Mixed States}}.
\newblock \emph{Communications on Pure and Applied Mathematics},
  69(12):2250--2303, December 2016.
  
\bibitem[BPSS21]{BPSS21}
Niels Benedikter, Marcello Porta, Benjamin Schlein, and Robert Seiringer.
\newblock Correlation Energy of a Weakly Interacting Fermi Gas with Large Interaction Potential.
\newblock \emph{arXiv:2106.13185 [math-ph]}, June 2021. 

\bibitem[BPS14c]{BPS14b} 
Niels Benedikter, Marcello Porta, and Benjamin Schlein.
\newblock Mean\textendash{{Field Evolution}} of {{Fermionic Systems}}.
\newblock \emph{Communications in Mathematical Physics}, 331(3):1087--1131,
  November 2014.
 
 \bibitem[BPS14]{BPS14}
  Niels Benedikter, Marcello Porta, and Benjamin Schlein.
  \newblock Mean--Field Dynamics of Fermions with Relativistic Dispersion.
  \newblock \emph{Journal of Mathematical Physics} 55:021901, 2014. 
  
  \bibitem[BPS16]{BPS16}
Niels Benedikter, Marcello Porta, and Benjamin Schlein.
\newblock \emph{Effective {{Evolution Equations}} from {{Quantum Dynamics}}}.
\newblock {{SpringerBriefs}} in {{Mathematical Physics}}. {Springer
  International Publishing}, 2016.
   

\bibitem[BP53]{BP53} David Bohm and  David Pines. \newblock A {{Collective Description}} of {{Electron Interactions}}: {{III}}.
  {{Coulomb Interactions}} in a {{Degenerate Electron Gas}}.
\newblock \emph{Physical Review}, 92(3):609--625, November 1953.

\bibitem[BS19]{BS19}
Christian Brennecke and Benjamin Schlein.
\newblock Gross\textendash{{Pitaevskii}} dynamics for
  {{Bose}}\textendash{{Einstein}} condensates.
\newblock \emph{Analysis \& PDE}, 12(6):1513--1596, 2019.

\bibitem[BSS18]{BSS18}
Niels Benedikter, J{\'e}r{\'e}my Sok, and Jan~Philip Solovej.
\newblock The {{Dirac}}\textendash{{Frenkel Principle}} for {{Reduced Density
  Matrices}}, and the {{Bogoliubov}}\textendash{}de {{Gennes Equations}}.
\newblock \emph{Annales Henri Poincar{\'e}}, 19(4):1167--1214, April 2018.


\bibitem[CG12]{CG12}
Michele Correggi and Alessandro Giuliani.
\newblock The {{Free Energy}} of the {{Quantum Heisenberg Ferromagnet}} at
  {{Large Spin}}.
\newblock \emph{Journal of Statistical Physics}, 149(2):234--245, October 2012.

\bibitem[CGS15]{CGS15}
Michele Correggi, Alessandro Giuliani, and Robert Seiringer.
\newblock Validity of the {{Spin}}-{{Wave Approximation}} for the {{Free
  Energy}} of the {{Heisenberg Ferromagnet}}.
\newblock \emph{Communications in Mathematical Physics}, 339(1):279--307,
  October 2015.
  
\bibitem[CHN21]{CHN21}
Martin Ravn Christiansen, Christian Hainzl, and Phan Thành Nam.
\newblock The Random Phase Approximation for Interacting Fermi Gases in
the Mean-Field Regime.
\newblock \emph{arXiv:2106.11161 [math-ph]}, June 2021.

\bibitem[CLL21]{CLL21}
Li Chen, Jinyeop Lee, and Matthew Liew.
\newblock Combined Mean-Field and Semiclassical Limits of Large Fermionic Systems.
\newblock \emph{Journal of Statistical Physics} 182:24, January 2021.

\bibitem[CLS21]{CLS21}
Jacky J. Chong, Laurent Lafleche, and Chiara Saffirio.
\newblock From many--body quantum dynamics to the Hartree--Fock and Vlasov equations with singular potentials.
\newblock \emph{arXiv:2103.10946 [math.AP]}, March 2021.

\bibitem[EESY04]{EESY04}
Alexander Elgart, László Erdős, Benjamin Schlein, and Horng-Tzer Yau.
\newblock Nonlinear Hartree equation as the mean field limit of weakly coupled fermions.
\newblock \emph{Journal de Mathématiques Pures et Appliquées},
83(10): 1241--1273, October 2004.

\bibitem[FGHP21]{FGHP21}
Marco Falconi, Emanuela L.~Giacomelli, Christian Hainzl, Marcello Porta.
\newblock The Dilute Fermi Gas via Bogoliubov Theory.
\newblock \emph{Annales Henri Poincaré}, 22:2283--2353, 2021.

\bibitem[FK11]{FK11}
Jürg Fröhlich and Antti Knowles.
\newblock A Microscopic Derivation of the Time--Dependent Hartree--Fock Equation with Coulomb Two--Body Interaction.
\newblock \emph{Journal of Statistical Physics} 145:23, 2011.

\bibitem[Gir62]{Gir62}
Marvin Girardeau.
\newblock Variational Method for the Quantum Statistics of Interacting Particles.
\newblock \emph{Journal of Mathematical Physics} 3:131--139, January 1962.

\bibitem[GM13]{GM13}Manoussos Grillakis and Matei Machedon.
\newblock Beyond mean field: On the role of pair excitations in the evolution of condensates.
\newblock \emph{Journal of Fixed Point Theory and Applications volume} 14, 91--111, September 2013.

\bibitem[GS94]{GS94}
Gian Michele Graf and Jan Philip Solovej.
\newblock A {{Correlation Estimate}} with {{Applications}} to {{Quantum
  Systems}} with {{Coulomb Interactions}}.
\newblock \emph{Reviews in Mathematical Physics}, 06(05a):977--997, January
  1994.
  
\bibitem[GS13]{GS13}
P.~Grech and R.~Seiringer.
\newblock The {{Excitation Spectrum}} for {{Weakly Interacting Bosons}} in a
  {{Trap}}.
\newblock \emph{Communications in Mathematical Physics}, 322(2):559--591,
  September 2013.
  
  \bibitem[Hea99]{Hea99}
David~R.~Heath--Brown.
\newblock \emph{Lattice points in the sphere}. In: Number Theory in Progress, pages 883--892. Berlin, Boston: De Gruyter (1999).
\newblock eISBN: 9783110285581. 
\newblock \url{https://doi.org/10.1515/9783110285581.883}     
  
  \bibitem[HPR20]{HPR20}
Christian Hainzl, Marcello Porta, and Felix Rexze.
\newblock On the {{Correlation Energy}} of {{Interacting Fermionic Systems}} in
  the {{Mean}}-{{Field Regime}}.
\newblock \emph{Communications in Mathematical Physics},  374:485--524(2020).

\bibitem[KB62]{KB62}
Albion J.~Kromminga, Mark Bolsterli.  
\newblock Perturbation Theory of Many--Boson Systems.
\newblock \emph{Physical Review} 128(6):2887, December 1962.

\bibitem[LNS15]{LNS15}
Mathieu Lewin, Phan Thành Nam, and Benjamin Schlein. Fluctuations around Hartree states in the mean--field regime. \emph{American Journal of Mathematics} 137(6):1613-1650, 2015. 

\bibitem[LS02]{LS02} Elliott H.\ Lieb and Jan Philip Solovej.
\newblock Ground state energy of the one-component charged Bosegas.
\newblock \emph{Communications in Mathematical Physics} 217(1): 127--163, 2001.\\
\newblock Errata: \emph{Communications in Mathematical Physics} 225:219–221, 2002.

\bibitem[LS04]{LS04} Elliott H.\ Lieb and Jan Philip Solovej.
\newblock Ground state energy of the two-component charged Bose gas. \newblock \emph{Communications in Mathematical Physics} 252:485--534, July 2004.
  
  \bibitem[ML65]{ML65}
Daniel~C.\ Mattis and Elliott~H.\ Lieb.
\newblock Exact {{Solution}} of a {{Many}}-{{Fermion System}} and {{Its
  Associated Boson Field}}.
\newblock \emph{Journal of Mathematical Physics}, 6(2):304--312, February 1965.

\bibitem[NN17]{NN17}
Phan Thành Nam and Marcin Napiórkowski. A note on the validity of Bogoliubov correction to mean--field dynamics.
\emph{Journal de Mathématiques Pures et Appliquées}, 108(5):662--688, 2017. 

\bibitem[NS19]{NS19}Marcin Napiórkowski and Robert Seiringer. Free energy asymptotics of the quantum Heisenberg spin chain.
\emph{Letters in Mathematical Physics}, 111:31, 2021.

\bibitem[NS81]{NS81}
Heide~{Narnhofer} and Geoffrey~L.~{Sewell}. {Vlasov hydrodynamics of a quantum
  mechanical model}. \emph{Communications in Mathematical Physics}, 79(1):9--24, 1981. 
  
\bibitem[PP16]{PP16}
Sören Petrat and Peter Pickl.
\newblock A New Method and a New Scaling for Deriving Fermionic Mean--Field Dynamics.
\newblock \emph{Mathematical Physics, Analysis and Geometry} 19:3, 2016. 
  
\bibitem[PRSS17]{PRSS17}
Marcello Porta, Simone Rademacher, Chiara Saffirio, and Benjamin Schlein.
\newblock Mean Field Evolution of Fermions with Coulomb Interaction. 
\newblock \emph{Journal of Statistical Physics} 166:1345--1364, January 2017.


\bibitem[Sei11]{Sei11} 
 Robert Seiringer.
 \newblock The  Excitation  Spectrum  for  Weakly  Interacting  Bosons.
 \newblock \emph{Communications in Mathematical Physics}, 306(2):565--578, May 2011.

\bibitem[Spo81]{Spo81}
H.~Spohn.
\newblock {On the {V}lasov hierarchy}.
\newblock \emph{Mathematical Methods in the Applied Sciences}, 3(4):445--455, 1981.
\end{thebibliography}
\end{document}